\documentclass[journal]{IEEEtran}

\usepackage{amsmath,amssymb,amsfonts,amsthm}
\usepackage{color}
\usepackage{graphicx}
\usepackage{epstopdf}
\usepackage{bbm}
\usepackage{algpseudocode} 
\usepackage{algorithm}
\usepackage{subcaption} 
\usepackage{url}

\newtheorem{lemma}{\textbf{Lemma}}

\newtheorem{definition}{\textbf{Definition}}
\newtheorem{thm}{\textbf{Theorem}}
\newtheorem{clm}{\textbf{Claim}}
\newtheorem{fact}{\textbf{Fact}}

\newtheorem{example}{\textbf{Example}}

\newtheorem{remark}{\textbf{Remark}}

\newtheorem{question}{\textbf{Question}}

\newcommand{\ie}{{i.e.}}
\newcommand{\eg}{{e.g.}}
\newcommand{\etal}{{et al.}}

\renewcommand{\S}{Section~}

\newcommand{\mc}[1]{\mathcal{#1}}

\renewcommand{\Pr}[1]{\ensuremath{\operatorname{\mathbf{Pr}}\left[#1\right]}}
\newcommand{\Ex}[1]{\ensuremath{\operatorname{\mathbf{E}}\left[#1\right]}}
\newcommand{\Var}[1]{\ensuremath{\operatorname{\mathbf{Var}}\left[#1\right]}}

\title{Storage, Communication, and Load Balancing Trade-off in Distributed Cache Networks}

\author{
Mahdi Jafari Siavoshani, \IEEEmembership{Member, IEEE,} 
Ali Pourmiri, Seyed Pooya Shariatpanahi\\%
\thanks{The authors' names appear in alphabetical order.}%
\thanks{This work has been presented in part at IPDPS 2017, \cite{JafPourShar_IPDPS17}.}
\thanks{
M. Jafari Siavoshani is with the Department of Computer Engineering, Sharif University of Technology, Tehran, Iran (email: mjafari@sharif.edu).}
\thanks{
A. Pourmiri is with the Department of Computer Engineering, University
of Isfahan, Isfahan, Iran (email: a.pourmiri@comp.ui.ac.ir).}
\thanks{
S. P. Shariatpanahi is with the School of Computer Science, Institute for
Research in Fundamental Sciences (IPM), Tehran, Iran (email: pooya@ipm.ir).}
}

\begin{document}
\maketitle

\begin{abstract}
We consider load balancing in a network of caching servers delivering contents to end users.
Randomized load balancing via the so-called \emph{power of two choices} is a well-known approach in parallel and distributed systems. 
In this framework, we investigate the tension between storage resources, communication cost, and load balancing performance.
To this end, we propose a randomized load balancing scheme which simultaneously considers cache size limitation and proximity in the server redirection process. 

In contrast to the classical power of two choices setup, since the memory limitation and the proximity constraint cause correlation in the server selection process, we may not benefit from the power of two choices.  
However, we prove that in certain regimes of problem parameters, our scheme results in the maximum load of order $\Theta(\log\log n)$ (here $n$ is the network size). This is an exponential improvement compared to the scheme which assigns each request to the nearest available replica. Interestingly, the extra communication cost incurred by our proposed scheme, compared to the nearest replica strategy, is small.
Furthermore, our extensive simulations show that the trade-off trend does not depend on the network topology and library popularity profile details.
\end{abstract}

\centerline{\textbf{Keywords}}
Randomized Algorithms, Distributed Caching Servers, Request Routing, Load Balancing,  Communication Cost, Balls-into-Bins, Content Delivery Networks.

\section{Introduction}\label{sec:Introduction}
\subsection{Problem Motivation}
Advancement of technology leads to the spread of smart multimedia-friendly communication devices to the masses which causes a rapid growth of demands for data communication \cite{Cisco_Report}. Although
Telcos have been spending hugely on telecommunication infrastructures,  they cannot keep up with this data demand explosion.
Caching predictable data in network off-peak hours, near end users, has been proposed as a promising solution to this challenge. This approach has been used extensively in content delivery networks (CDNs) such as Akamai, Azure, Amazon CloudFront, etc. \cite{Zhang2013}, \cite{NygrenSS10}, and mobile video delivery  \cite{GolrezaeiHelper}. In this approach, a \emph{cache network} is usually referred to as a set of caching servers that are connected over a network, giving content delivery service to end users.

As a schematic view of a typical cache network see Figure~\ref{fig:Setup_Cloud}. This figure shows caching servers connected through a backhaul network. These servers are responsible for delivering contents requested by customers.  Each server should assign the demand to a responding server (which could be itself) via an assignment strategy. In every cache network we have three critical parameters, namely, \emph{Storage Resource}, \emph{Communication Cost}, and \emph{Network Imbalance}. Storage, or memory characterizes what percentage of the total content library can be cached at each server. Communication cost is the amount of data transferred inside the backhaul network to satisfy content requests. Finally, network imbalance characterizes how uniformly the requested contents' loads are distributed among different responding servers. This is usually measured by comparing the load of the busiest server with the average load of all servers after request assignments. Every request assignment strategy, in fact, leads to a trade-off between these three quantities. 

\begin{figure}
\begin{center}
\includegraphics[scale=0.45]{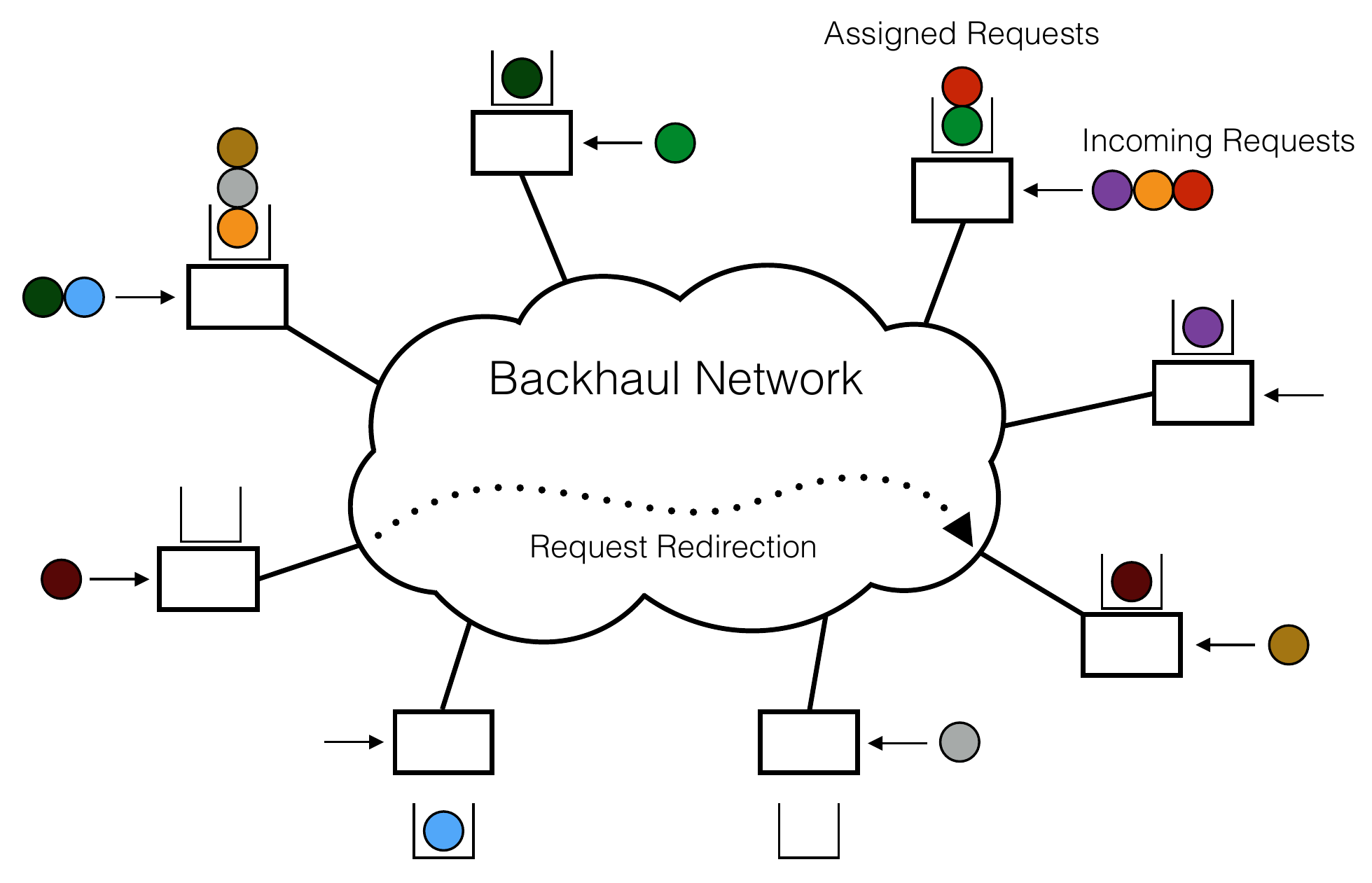}
\end{center}
\caption{A general distributed cache network.}
\label{fig:Setup_Cloud}
\end{figure}


From a practical viewpoint, we are interested in \emph{distributed} server selection strategies which are scalable in large networks. Randomized load balancing via the so-called ``power of two choices'' is a well-investigated paradigm in parallel and distributed settings \cite{ABKU99,Mitzenmacher01,AdlerCMR98,LenzenW16}. In this approach, upon arrival of a request, the corresponding user will query about the current loads of two independently at random chosen  servers, and then allocates the request to the least loaded server.
Considering only the load balancing perspective, Berenbrink \etal, \cite{BCSV06}, showed that in this scheme after allocating  $m$ balls (requests, tasks, etc.) to $n$ bins (servers, machines, etc.) the   maximum number of balls assigned to any bin, called {\it maximum load}, is at most $m/n+O(\log\log n)$ with high probability. This  only deviates $O(\log\log n)$ from the average load and the deviation  depends on the number of servers. 

Although the power of two choices strategy addresses the load balancing issue in a distributed manner, it does not consider the role of two other important quantities, namely, memory and communication cost. Our goal in this paper is to extend the power of two choices framework in order to characterize the trade-off between these three quantities. Our results show that the maximum load, communication cost, and servers' memory are three entangled parameters. Thus, the previous studies are not sufficient for designing a practical load balancing strategy in cache networks.

\subsection{Problem Setting and Our Contributions}


In this paper, we consider a general cache network model that entails basic characteristics of many practical scenarios. We consider a network of $n$ servers and a library of size $K$ files. Each server can cache $M$ files in network low-traffic hours. Let us assume a popularity distribution $\mathcal{P}=\{p_1,\dots,p_K\}$ on the library. We assume that the cache content placement at each server is proportional to this popularity distribution.
In high-traffic hours there are $n$ sequential file requests, from the library, distributed among servers uniformly at random. Every server either serves its requests or redirects them (via an assignment scheme) to other nodes which have cached the files. We define the \emph{maximum load} $L$ of an assignment scheme as the maximum number of allocations to any single server after assigning all requests. The \emph{communication cost} $C$ is the average number of hops required to deliver the requested file to its request origin. 

As the baseline assignment scheme, we consider that each request arrived at every server is dispatched to the nearest file replica. This scheme results in the minimum communication cost, while ignoring maximum load of servers. We show that, for a grid topology and every constant $0<\alpha<1/2$, if  $K=n$, $M=n^{\alpha}$, and $\mathcal{P}$ is a uniform distribution, this scheme will result in the maximum load $L$ in the interval $[\Omega(\log n/\log\log n), O(\log n)]$ with high probability\footnote{With high probability refers to an event that happens with probability $1-1/n^c$, for some constant $c>0$.} (w.h.p.).  Moreover, for
every constant $0<\epsilon<1$, if 
 $K=n^{1-\epsilon}$ and $M=\Theta(1)$, then  the maximum load is $\Theta(\log n)$ w.h.p.
We also investigate the communication cost incurred in this scheme for Uniform and Zipf popularity distributions. In particular, we derive the communication cost $C$ of order $\Theta(\sqrt{K/M})$ for the Uniform distribution in a grid topology.

In contrast, we propose a new scheme which considers memory, maximum load, and communication cost, simultaneously. For each request, this scheme chooses two random candidate servers that have cached the request while putting a constraint on their distance $r$ to the requesting node (\ie, the proximity constraint). Due to cache size limitation and the proximity constraint, current results in the balanced allocation literature cannot be carried over to our setting. Basically, we show that here the two chosen servers will become correlated and this might diminish the power of two choices. Since this correlation arises from both memory limitation and proximity constraint, the main challenge we address in this paper is characterizing the regimes where we can benefit from the power of two choices and at the same time have a low communication cost.

In particular, suppose $0<\alpha, \beta<1/2$ be two constants and
	let  $K=n$, $M=n^\alpha$,  $r=n^{\beta}$, and $\mathcal{P}$ be a Uniform distribution. Then, for grid topology, provided $\alpha+2\beta\geq 1+2(\log\log n/\log n)$, the maximum load is $\Theta(\log\log n)$ w.h.p., and the communication cost is $\Theta(r)$. Therefore, if we set $\beta={\frac{1-\alpha}{2}}+\log\log n/\log n$ we achieve the power of two choices  with the communication cost of order $\Theta(r)=\Theta\left(n^{\frac{1-\alpha}{2}} \log n\right)$. This  communication cost is only $\log n$ factor above the communication cost achieved by the nearest replica strategy, which is $\Theta(\sqrt{K/M})=\Theta({n^{\frac{1-\alpha}{2}}})$. Figure~\ref{fig:alpha_beta_tradeoff} shows the region of parameters $\alpha$ and $\beta$ where the power of two choices is asymptotically achievable. It should be noted that while our theoretical results are derived for grid networks, the main reason for assuming a grid is presentation clarity and the results can be extended to other topologies.

\begin{figure}
\begin{center}
\includegraphics[scale=0.60]{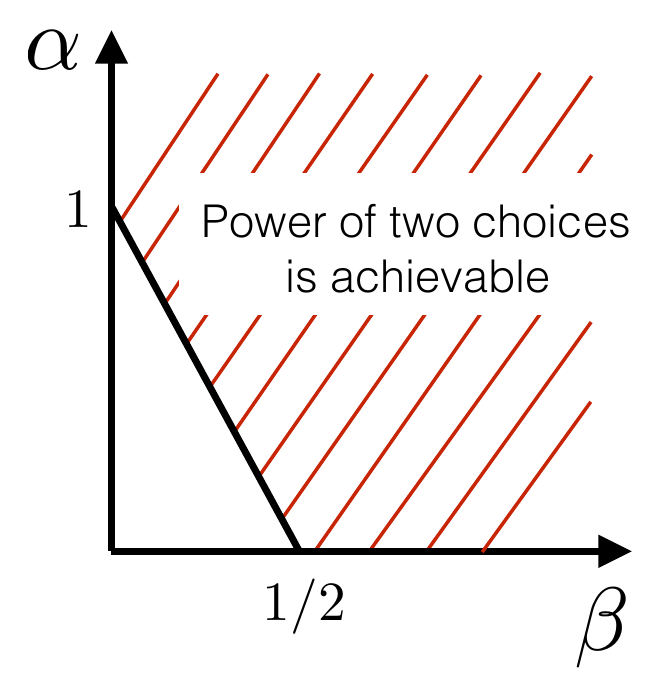}
\end{center}
\caption{Suppose $M=n^\alpha$ and $r=n^\beta$. Then the shaded area shows the region where the power of two choices is asymptotically achievable (for more details, refer to Theorem~\ref{main:twochoice}).}
\label{fig:alpha_beta_tradeoff}
\end{figure}

While our theoretical results are derived for large networks (\ie, asymptotic analysis), our simulation results show their validity even for finite sized networks. Also, in simulations we investigate the problem in more diverse settings such as considering other network topologies.

\subsection{Related Work}\label{sec:RelatedWork}
Load balancing has been the focus of many papers on cache networks \cite{CDN_Peng_2004, RoussopoulosB06, LeconteLM12}, among which distributed approaches have attracted a lot of attention (\eg, see \cite{ManfrediOR13}, \cite{AdlerCMR98}, and \cite{XiaAYL15}). Randomized load balancing via the power of two choices, is a popular approach in this direction \cite{Mitzenmacher01}. Chen et al. \cite{ChenLPCCSA05} consider the two choices selection process, where the second choice is the next neighbor of the first choice.  In \cite{XiaDH07}, Xia et al. use the length of common prefix (LCP)-based replication to arrive at a recursive balls and bins problem. In \cite{ChenLPCCSA05} and \cite{XiaDH07}, the authors benefit from the metaphor of power of two choices to design algorithms for randomized load balancing. In contrast to these works, we follow a theoretical approach to derive provable results for cache networks. 

Theoretical works investigating the power of two choices in cache networks all just consider the role of two parameters among \emph{Memory Resource},  \emph{Load Balancing}, and \emph{Communication Cost}. These works can be summarized in three categories as follow.

\emph{Memory Resource} vs. \emph{Communication Cost} trade-off has been investigated in many works such as \cite{GolrezaeiHelper}, \cite{Meyerson-2001}, \cite{Borst-2010}, \cite{Maddah-AliN14}, and \cite{Gitzenis-2013}. Non of these papers have considered the load balancing issue, in cache networks which is very important in practice.

\emph{Load Balancing} vs. \emph{Memory Resource} trade-off has been investigated in \cite{ShahV15} and \cite{LiuST16}. In \cite{ShahV15} the authors consider the supermarket model for performance evaluation of CDNs. Although the work \cite{ShahV15} considers the memory limitation into account, it does not consider the proximity principle which is a central issue in our paper. 
Liu et al. \cite{LiuST16} study the setting where the clients compare the servers in terms of hit-rate (for web applications), or bit-rate (for video applications) to choose their favourite ones. Their setup and objectives are different from those we consider here. Moreover, they have not considered the effect of their randomized load balancing scheme on communication cost.

\emph{Communication Cost} vs. \emph{Load Balancing} trade-off has been investigated in \cite{PathanVB08, TangTW14}, \cite{StanojevicS09}, \cite{BBFN12}, \cite{KP06}, \cite{God08}, and \cite{Pou16}, without considering the effect of cache size limitation. Although the works \cite{PathanVB08, TangTW14}, and \cite{StanojevicS09} have mentioned this trade-off, non of them provides a rigorous analysis. 

In contrast to the standard balls and bins model, the works \cite{BBFN12}, \cite{KP06}, \cite{God08}, and \cite{Pou16} introduced the effect of proximity constraint to the ball and bins framework. In the standard balls and bins model, each ball (request) picks two bins (servers) independently and uniformly at random and it is then allocated to the one with lesser load \cite{ABKU99}. However, in many settings, selecting any two random servers might be infeasible or costly. In other words, this proximity constraint translates to a \emph{correlation} between the two choices, \ie, the balls and bins model with \emph{related choices}.

The most related work to our paper is \cite{KP06}. Kenthapadi and Panigrahi \cite{KP06} proposed a model where $n$ bins are connected as a $d$-regular graph. Corresponding to each ball, a node is chosen uniformly at random as the first candidate. Then, one of its neighbours is chosen uniformly at random as the second candidate and the ball is allocated to the one with the minimum load. 
Under this assumption, they proved that if the graph is sufficiently dense (\ie, the average degree is $n^{\Omega(\log\log n/\log n)}$), then after allocating  $n$ balls, the maximum load is $\Theta(\log\log n)$ w.h.p.

Although the model used in \cite{KP06} considers the proximity principle by assigning each request to the origin neighbors, it cannot be directly applied to our cache network setup. First, they do not consider multi-hop communication, while in practice the communication is done in a multi-hop fashion. Second, the above model cannot accommodate the cache size limitation of servers. Cache size limitation introduces the notion of cache content placement which should be based on the popularity profile. In addition, this limitation will introduce a new source of correlation between choices which is not considered in \cite{KP06}.

The organization of the paper is as follows. In Section~\ref{sec:Notation_ProblemSetting}, we present our notation and problem setup. Then, in Section \ref{sec:Nearest_Replica_Strategy} the \emph{nearest replica strategy}, is investigated as the baseline scheme. In Section \ref{sec:Proximity-Aware_Two_Choices_Strategy}, we propose and analyze the \emph{proximity-aware two choices strategy}, which at the same time considers memory limitation, proximity constraint, and benefits from the power of two choices. In Section \ref{sec:Simulations} the performance of these two schemes are investigated via extensive simulations. Finally, our discussions and concluding remarks are presented in  \S\ref{sec:Discussion_Remarks}.

\section{Notation and Problem Setting}\label{sec:Notation_ProblemSetting}

\subsection{Notation}
Throughout the paper, with high probability refers to an event that happens with probability $1-1/n^c$, for some constant $c>0$. Let $G=(V,E)$ be a graph with vertex set $V$ and edge set $E$ where $e(G):= |E|$. For $u\in V$ let $d(u)$ denote for the degree of $u$ in $G$. For every pair of nodes $u,v\in V$,   $d_G(u, v)$ denotes the length of a shortest path from $u$ to $v$ in $G$. The neighborhood of $u$ at distance $r$ is defined as    
\[
  B_r(u) := \left\{v : d_G(u, v)\le r ~ {\text {and}}~~ v\in V(G) \right\}.
\]
Finally, we use $\mathrm{Po}(\lambda)$ to denote for the Poisson distribution with parameter $\lambda$.

\subsection{Problem Setting}

We consider a cache network consisting of $n$ caching servers (also called cache-enabled nodes) and edges connecting neighboring servers forming a $\sqrt{n} \times \sqrt{n}$ grid. Direct communication is possible only between adjacent nodes, and other communications should be carried out in a multi-hop fashion.

\begin{remark}
	Throughout the paper for the sake of presentation clarity we may consider a torus with $n$ nodes. This helps to avoid boundary effects of grid 
	and all the asymptotic results hold for the grid as well.
\end{remark}

Suppose that the cache network is responsible for handling a library of $K$ files $\mathcal{W}=\{W_1,\dots,W_K\}$, whereas the popularity profile follows a known distribution $\mathcal{P}=\{p_1,\dots,p_K\}$.

The network operates in two phases, namely, \emph{cache content placement} and \emph{content delivery}. In the cache content placement phase each node caches $M\le K$ files randomly from the library according to their popularity distribution $\mathcal{P}=\{p_1,\dots,p_K\}$ with replacement, independent of other nodes.
Also note that, throughout the paper we assume that $M\ll K$, unless otherwise stated.


Consider a time block during which $n$ files are requested from the servers sequentially that are chosen uniformly at random. Let $D_i$ denote the number of requests (demands) arrived at server $i$. Then for large $n$ we have $D_i \sim \mathrm{Po(1)}$ for all $1\le i\le n$.


For library popularity profile $\mc{P}$, we consider two probability distributions, namely, Uniform and Zipf with parameter $\gamma$. In the  Uniform distribution we have
\begin{equation*}
p_i=\frac{1}{K},\quad i=1,\dots,K,
\end{equation*}
which considers equal popularity for all the files. In Zipf distribution the request probability of the $i$-th popular file is inversely proportional to its rank as follows
\begin{equation*}
p_i=\frac{1/i^\gamma}{\sum\limits_{j=1}^{K}1/j^\gamma},\quad i=1,\dots,K,
\end{equation*}
which has been confirmed to be the case in many practical applications \cite{Zipf1_99,Zipf2_07}.

For any given cache content placement, an assignment strategy determines how each request is mapped to a server. 
Let $T_i$ denote the number of requests assigned to server $i$ at the end of mapping process.

Now, for each strategy we define the following metrics.
\begin{definition}[Communication Cost and Maximum Load]
~
\begin{itemize}
\item The \emph{communication cost} of a strategy is the average number of hops between the requesting node and the serving node, denoted by $C$. 
\item The \emph{maximum load} of a strategy is the maximum number of requests assigned to a single node, denoted by $L=\max_{1\le i\le n} T_i$. 
\end{itemize}
\end{definition}

\section{Nearest Replica Strategy}\label{sec:Nearest_Replica_Strategy}
The simplest strategy for assigning requests to servers is to allocate each request to the nearest node that has cached the file. This strategy, formally defined below, leads to the minimum communication cost, while does not try to reduce maximum load.

\begin{definition} [Strategy I: Nearest Replica Strategy]
In this strategy each request is assigned to the nearest node --in the sense of the graph shortest path distance-- which has cached the requested file. If there are multiple choices ties are broken randomly.
\end{definition}

Consider the set of nodes that have cached file $W_j$, say $S_j$. According to Strategy I, each demand from node $u$ for file $W_j$ will be served by $\arg\min_{v \in S_j} d_G(u,v)$. This induces a Voronoi Tessellation on the torus corresponding to file $W_j$ which we denote by $\mathcal{V}_j$. Then, alternatively, we can define Strategy I as assigning each request of file $W_j$ to the corresponding Voronoi cell center. 

In order to analyze the maximum load imposed on each node, we should investigate the size of such Voronoi regions. The following Lemma is in this direction.

\begin{lemma}\label{WindMillLemma}
Under the Uniform popularity distribution,  the maximum cell size (number of nodes inside each cell) of $\mathcal{V}_j$, $1\le j\le K$, is at most $O\left(K \log n / M\right)$ w.h.p. In particular, every Voronoi cell centered at any node is contained in a sub-grid of size $r\times r$  with $r=O\left(\sqrt{K \log n / M}\right)$.
	 Furthermore, if $K=n^{1-\epsilon}$, for some  constant $0<\epsilon<1$, and $M=\Theta(1)$, then there exists a Voronoi cell of size $\Theta\left(K \log n / M\right)$ w.h.p.
\end{lemma}

\begin{proof}[Proof of Lemma~\ref{WindMillLemma}]
\textbf{Upper Bound -- } 
Fix a node $u$ and w.l.o.g. assume that 
$u$ is denoted by pair $(0,0)$ in the torus.
With respect to $u$ and some positive number $r>0$  define four areas as follows
\begin{align*}
A_1(u)&:=\{(x, y) : 0\le y \le x/2  \text{ and }  (x, y)\in B_{r}(u)\},\\
A_2(u)&:=\{(x, y): 0\le -x\le y/2
\text{ and }  (x, y)\in B_{r}(u)
\},\\
A_3(u)&:=\{(x, y): 0\le -y\le -x/2 
\text{ and }  (x, y)\in B_{r}(u)
\},\\
A_4(u)&:=\{(x, y): 0\le x\leq -y/2
\text{ and }  (x, y)\in B_{r}(u)
\},
\end{align*}
which are shown in Fig. \ref{Upper}. It is easy to see that all four areas have the same size, that is 
\begin{align}\label{eq:lower}
|A_1(u)| &= \sum_{y=0}^{\lfloor r/3 \rfloor}\sum_{x=2y}^{r-y}1 \nonumber\\
&=\sum_{y=0}^{\lfloor r/3\rfloor}(r-3y+1) \nonumber\\
&\geq \sum_{y=0}^{\lfloor r/3\rfloor} 3y  \nonumber\\
&\ge r^2/8.
\end{align}

Let us fix some arbitrary  $1\le j\le K$ 
and for every node $u$ define indicator random variable  $X_{u, j}$ taking value $1$ if  $u$ has cached file $W_j$ and there is no node in $A_1(u)$ that has cached file $W_j$, and $0$ otherwise. Then,
\begin{align*}
\Pr{X_{u, j}=1}=
\left(1-\left(1-\frac{1}{K}\right)^M\right)\left(1-\frac{1}{K}\right)^{M(|A_1(u)|-1)},
\end{align*}
where the first term determines the probability that $u$ caches $W_j$ and the second one determines the probability that nodes in $A_1(u)\setminus\{ u\}$ do not cache $W_j$. 
By setting $r=5\sqrt{K\log n/M}$ and applying Inequality (\ref{eq:lower}) we have, 
\begin{align}\label{approx1}
\left(1-\frac{1}{K}\right)^{M(|A_1(u)|-1)} &= \mathrm{e}^{-\frac{25\log n}{8}}(1+o(1)) \nonumber\\
&=O(n^{-3}),
\end{align}
where it follows from \mbox{$1-1/K= \mathrm{e}^{-1/K}(1+o(1))$} and $M/K=o(1)$. Moreover, by using the approximation \mbox{$1-(1-1/K)^M=M(1+o(1))/K$}, we have
\[
\Pr{X_{u,j}=1}\le  \frac{M(1+o(1))}{K\cdot n^3}.
\] 

Therefore applying the union bound over all $n$ nodes and $K$ files implies  that w.h.p. for every $u$ there exists at least one node in $A_1(u)$ which shares a common file with $u$, supported that we choose $r= 5\sqrt{K\log n/M}$. We can similarly prove the same argument for $A_2(u)$, $A_3(u)$ and $A_4(u)$. 

Suppose that $u$ has cached file $W_j$, and we want to find an upper bound for the size of the Voronoi cell centered at $u$ corresponding to $W_j$. In order to do this let us define $v^1=(v^1_x, v^1_y)\in A_1(u)$ to be the nearest node to $u$ with file $W_j$. Similarly define $v^i\in A_i(u)$, $2\le i\le 4$. W.l.o.g.
assume that $u$  is the origin and the nodes are in a $xy$-coordinate system. Define
 \[
B:=\{(x,y) : v^3_x\le x\le v^1_x \text{ and } v^4_y\le  y\le v^2_y\}.
\]
Now we show that the Voronoi cell of $u$ is contained in $B$, and thus the size of $B$ is an upper bound to the size of the Voronoi cell.
Consider Fig. \ref{Upper}. Let us consider  node  $w$ in the complement of $B$ with $w_y>v^2_y$ and $w_x>0$. Assume that  $P_{uw}$ is a shortest path from $w$ to $u$ that passes node $(0, v^2_y)$. By definition of $A_2(u)$, we know  the length of a shortest  path from    $(0, v^2_y)$ to $(v^2_x, v^2_y)$ is $|v^2_x|\leq v^2_y/2$. This shows that $w$ is closer to $v^2$ than $u$, and by definition it  does not belong to the Voronoi cell centered at $u$. We similarly can show that each arbitrary node $w\in B^c$ is closer to either $v^i$'s rather than $u$.  Since we arbitrarily choose $u$ and $1\le j\le K$, there is  sub-grid $B$ that contains every Voronoi cell in $\mathcal{V}_j$, centered at any given $u$.  
So the size of any Voronoi cell centered at an arbitrary node $u$ is bounded from above by $|B|$, which is at most $4r^2=O(K\log n/M)$.

\begin{figure}
\begin{center}
\includegraphics[width=0.5\textwidth]{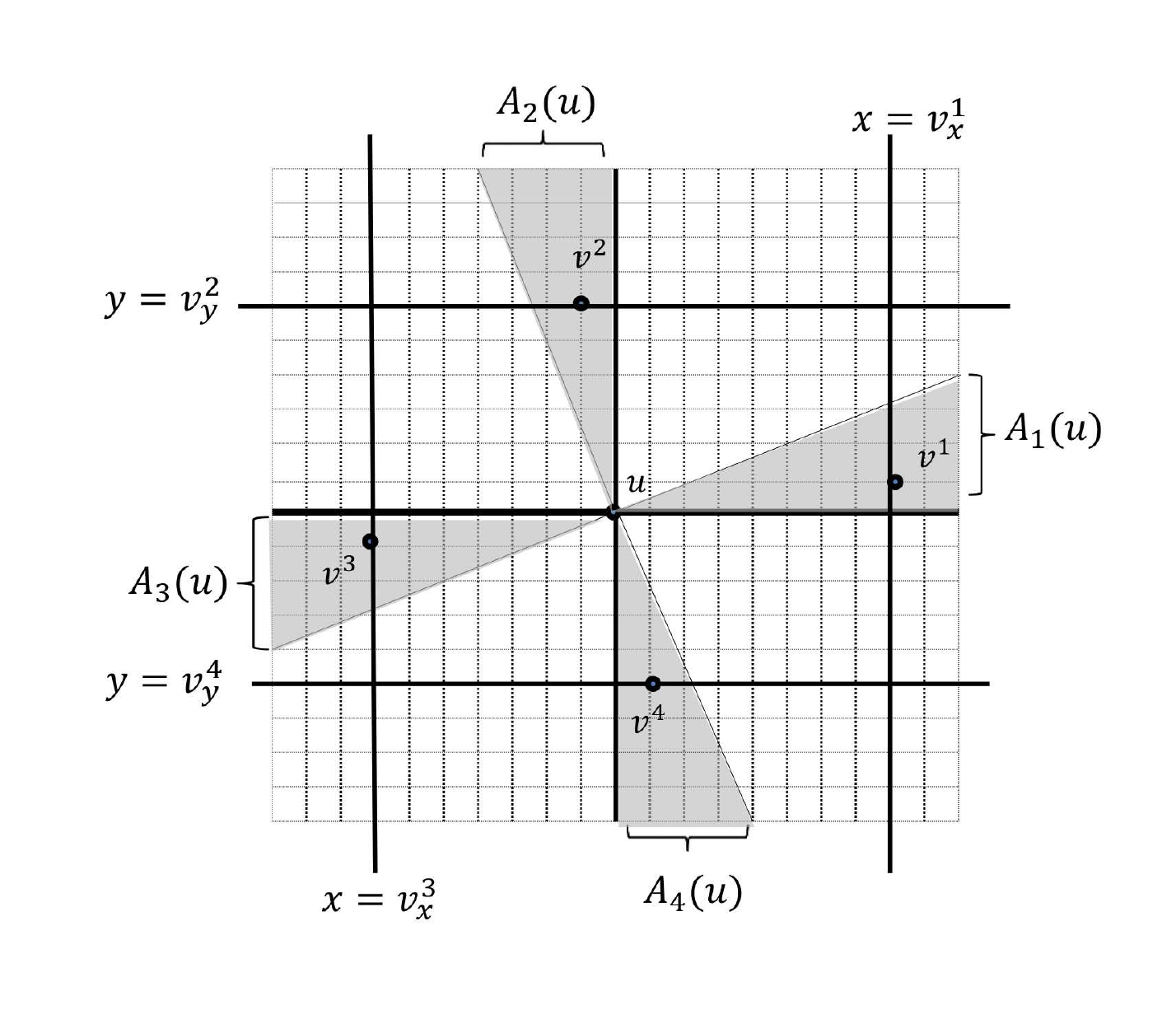}
\end{center}
\caption{Demonstration of regions $A_1(u),\ldots,A_4(u)$ used in the upper bound proof of Lemma~\ref{WindMillLemma}.}\label{Upper}
\end{figure}

\textbf{Lower Bound --} Let us define indicator random variable $Y_{u, j}$ for every $u$ and some fixed  $j$ taking value $1$ if $u$ has cached file $W_j$ and there is no  $v\in B_{r}(u)$ that has cached file $W_j$, and $0$ otherwise.	
Note that $|B_r(u)\setminus\{u\}|=2r(r+1)$.
Then we have  
\[
\Pr{Y_{u,j}=1}=\left(1-\left(1-\frac{1}{K}\right)^M\right)\left(1-\frac{1}{K}\right)^{M[2r(r+1)]}.
\]
By setting $r=\sqrt{{\epsilon\cdot K\cdot \log n /4M}}$ and using similar approximations used in (\ref{approx1}) we have
\[
p \triangleq \Pr{Y_{u,j}=1}=\frac{M(1+o(1))}{K\cdot n^{\epsilon/2}}.
\]
Let $Y_j=\sum_{u}  Y_{u, j}$. Then, we have the following claim.

\begin{clm}
For every $j$ we have $Y_j\ge 1$ with probability $1-o(1)$.
\end{clm}

This claim shows that there exists at least a Voronoi cell of size $\Theta(r^2)=\Theta(K \log n /M)$ which concludes the proof.

Now, in order to prove the claim note that
\[
\Ex{Y_j}=\sum_{u}\Ex{Y_{u,j}}=n\cdot\frac{M(1+o(1))}{K\cdot n^{\epsilon/2}}=Mn^{\epsilon/2}(1+o(1)).
\]
Also, we know that  
\begin{align}\label{var:main}
\Var{Y_j} &= \sum_{u, v}\mathrm{Cov}(Y_{u,j}, Y_{v,j}) \nonumber\\
&=\sum_{u, v}\Ex{Y_{u,j} Y_{v,j}} - \Ex{Y_{u,j}}\Ex{Y_{v,j}} \nonumber\\
&=\sum_{u, v} \Big( \Pr{Y_{u,j}=1, Y_{v,j}=1} \nonumber\\
&\quad\quad\quad -\Pr{Y_{u,j}=1}\cdot \Pr{Y_{v,j}=1}\Big),
\end{align}
where the last equality holds because 
$Y_{u, j}$'s are indicator random variables. 
It is easy to see that for every $u$
and $v$ with $d_G(u, v)>2r$, $\mathrm{Cov}(Y_{u,j}, Y_{v, j})=0$ as cache content placement at different nodes are independent processes. So we only consider pairs $u$ and $v$, with $d_G(u, v)\le 2r$.  Then for each pair of nodes three following cases should be considered:
\begin{itemize}
	\item $u=v$:
	In this case we have 
	\begin{align}\label{var:fisrt}
	\Pr{Y_{u,j}=1, Y_{v,j}=1}-\Pr{Y_{u,j}=1}\cdot \Pr{Y_{v,j}=1}\nonumber \\
	=\Pr{Y_{u,j}=1}-p^2=p(1-p).
	\end{align}
	
	\item $0<d_G(u, v)\le r$:
	By definition of indicator random variables $Y_{u, j}$'s, we have
	\begin{align}\label{var:sec}
	&\Pr{Y_{u,j}=1, Y_{v,j}=1}-\Pr{Y_{u,j}=1}\cdot \Pr{Y_{v,j}=1}\nonumber\\
	&=0-p^2. 
	\end{align}
	\item $r<d_G(u, v)\leq 2r$:
	In this case we have  	
	\begin{align}\label{var:thi}
	&\Pr{Y_{u,j}=1, Y_{v,j}=1}-\Pr{Y_{u,j}=1}\cdot \Pr{Y_{v,j}=1} \nonumber\\
	&=\Pr{Y_{u,j}=1|Y_{v,j}=1}\Pr{Y_{v,j}=1}-p^2 \nonumber \\ 
	&\leq \frac{M(1+o(1))}{K}p-p^2\leq \frac{2M}{K}p. 
	\end{align}
\end{itemize}

Now let us split the summation (\ref{var:main}) based on $d_G(u,v)$ as follows
\begin{align*}
\Var{Y_j} &=\sum_{u}\mathrm{Cov}(Y_{u,j},Y_{u,j})\\
&\quad+\sum_{u} \sum_{v: 0 < d_G(u, v)\leq r}\mathrm{Cov}(Y_{u,j}, Y_{v,j}) \\
&\quad+ \sum_{u} \sum_{v: r<d_G(u, v)\leq 2r}\mathrm{Cov}(Y_{u,j}, Y_{v,j}).
\end{align*}

Applying results (\ref{var:fisrt}-\ref{var:thi}) yields
\begin{align*}
\Var{Y_j} &\leq np(1-p)-n|B_r(u)|p^2+n|B_{2r}(u)|\frac{2Mp}{K} \nonumber\\
&\leq np+ 4nr(2r+1)\frac{2Mp}{K} \nonumber\\
&\leq np+6\epsilon np \log n \nonumber\\
&\leq 7np\log n,  
\end{align*}
where we use the fact that 
$4r(2r+1)\le 9r^2\leq 3\epsilon K \log n$.
Applying Chebychev's inequality  leads to 
\begin{align*}
\Pr{|Y_j-\Ex{Y_j}|\ge \Ex{Y_j}/2}&\le \frac{4\Var{Y_j}}{\Ex{Y_j}^2}\\
&\le \frac{28 np\log n}{n^2p^2}\\
&=\frac{28\log n}{np}=o(1).
\end{align*}
Therefore, $Y_j$ concentrates around its mean, \ie, $np = \Theta(n^{\epsilon/2})$, which proves the claim.
\end{proof}

Now, we are ready to present our main results for this section which characterize the maximum load of Strategy I, for two different parameter regimes, in Theorems~\ref{thm:imblance}~and~\ref{thm:Strategy_I_Large}.

\begin{thm}\label{thm:imblance}
Suppose that $K=n^{1-\epsilon}$, for some constant $0< \epsilon <1$, and $M=\Theta(1)$. Then, under Uniform distribution $\mathcal{P}$, Strategy I achieves maximum load of $L=\Theta(\log n)$ w.h.p.
\end{thm}

\begin{proof}
	Consider node $u$ which has cached a set of distinct files, say $S$, with $|S|\le M$. Applying Lemma \ref{WindMillLemma} shows that all  Voronoi cells centered at $u$ corresponding to cached files at $u$ are contained in a sub-grid of size  at most $O(K \log n /M)$  w.h.p. Also in each round, every arbitrary node  requests for a file in $S$ with probability  $|S|/nK\le M/nK$, as each request  randomly chooses its origin and type. Hence, by union bound, a node in the sub-grid may request for a file in $S$  with probability at most $O(K\log n/M)\cdot (M/nK)=O(\log n/n)$. Since there are $n$ requests, the expected number of requests  imposed to node $u$ is $O(\log n)$. Now using a Chernoff bound (e.g., see Appendix \ref{app:bounds}) shows that w.h.p. $u$ has to handle   at most $O(\log n)$ requests.

On the other hand, to establish a lower bound on the maximum load we proceed as follows. Lemma \ref{WindMillLemma} shows that there exits a Voronoi cell in $\mathcal{V}_j$, for some $j$, such that the center node should handle the requests of at least $\Theta(K \log n /M)$ nodes w.h.p. Also each node in the cell may  request for file $W_j$ with probability $1/nK$.  So on average there are $\Theta(\log n/M)$ requests imposed on the cell center. Similarly, by  a Chernoff bound, one can see that this node experiences the load  $\Theta(\log n /M)$, which concludes the proof for constant $M$.
\end{proof}

\begin{remark}
It should be noted that the same result of $\Theta(\log n)$ for the maximum load can also be proved for the Zipf distribution. That is because the content placement distribution is chosen proportional to the file popularity distribution $\mathcal{P}$, and consequently this result is insensitive to $\mathcal{P}$. However, the proof involves lengthy technical discussions which we omit in this paper. For further numerical investigation on this remark refer to \S\ref{sec:Simulations}. 
\end{remark}

\begin{thm}\label{thm:Strategy_I_Large}
	Suppose that $K=n$ and $M=n^{\alpha}$, for some $0< \alpha<1/2$. Then, under the Uniform distribution, the maximum load is in the interval $[\Omega(\log n/\log\log n), O(\log n)]$ w.h.p. 
\end{thm}

\begin{proof}[Proof of Theorem \ref{thm:Strategy_I_Large}]
	To establish  upper bound $O(\log n)$ for the maximum load, we follow the first part of proof of Theorem \ref{thm:imblance}.
	To obtain a lower bound, consider an arbitrary server $u$ that has cached file set $S$ with  $s$ distinct files. Note that by  Lemma 
	\ref{lem:GoodnessProperty}, we have for every node $u$, $s=\Theta(M)$ with high probability.
	Let us define the indicator random variable
	$X_{u, j}$, $W_j\in S$, taking $1$ if the nearest replica of $W_j$ is outside of $B_r(u)$, where $r=\sqrt{K/2M}$  and zero otherwise. 
	It is easy to see that $X_{u,j}$'s are correlated. For example, consider the set of files $T=\{W_{j_1}, W_{j_2},\ldots, W_{j_t}\}\subset S$, where $X_{u,j'}=1$ for every $W_{j'}\in T$. Then conditioned on this event, each node in $B_r(u)$ has cached files from a subset of the library of size $K-|T|$.
	Then probability that a node in $B_r(u)$ caches $W_j$ is at most $M/(K-t)$.
	Hence,  for every $W_j\in S$,
	\begin{align*}
	&\Pr{X_{u,j}=1 |\{X_{u, j'}=x_{u,j'}, W_{j'}\in S\setminus\{W_j\}\}}\\
	&\geq\left(1-\frac{M}{K-\sum_{ W_{j'}\in S\setminus\{W_j\}}x_{u, j'}}\right)^{2r(r+1)}\\
	&\geq \left(1-\frac{M}{K-M+1}\right)^{2r(r+1)}=\mathrm{e}^{-\Omega(1)}=p,
	\end{align*}
	where $|B_r(u)\setminus\{u\}|=2r(r+1)=\Theta(K/M)$ and hence $p$ is a constant.
	Let $Z=\sum_{W_j\in S}X_{u,j}$ and then $\Ex{Z}\geq s\cdot p$. Using a Chernoff bound for moderately correlated indicator random variables (e.g., see Lemma \ref{mod-chernoff}) implies that 
	\[
	\Pr{Z<sp/2}=o(1/n^2).
	\] 
	Therefore $B_r(u)$ does not contain any replica of at least $p/2$ fraction of files cached at $u$, namely $S'$. Using the union bound over all nodes we deduce the similar statement for every node with probability at least $1-o(1/n)$.
	Therefore, for every $u$ we have,
	\[
	\Pr{\text{$u$ severs a request}} \geq \frac{|B_{r/2}(u)|}{n}\cdot \frac{|S'|}{K}=\Omega(1/n)
	\]
	where it follows from $|S'|=\Theta(M)$, $|B_{r/2}(u)|=\Theta(K/M)$.
	Since there are $n$ requests, it is easy to see that   the load of each server is bounded from below by a  Poisson distribution $\mathrm{Po}(c)$, where $c$ is a constant. 
	On the other hand, it is known that (e.g., see \cite{Devroye85}) the maximum number taken by  $n$ Poisson distribution $\mathrm{Po}(c)$ is $\Omega(\log n/\log\log n)$ w.h.p. and hence the lower bound is proved.
\end{proof}

Next, we investigate the communication cost of Strategy~I in the following  theorem.

\begin{thm}\label{thm:Strategy_I_CommCost}
	Under the Uniform popularity distribution, Strategy I achieves the communication cost $C=\Theta(\sqrt{K/M})$, for every $M\ll K$. Under Zipf popularity distribution with $M=\Theta(1)$, it achieves
	\begin{equation}
		C = \left\{
		\begin{array}{llll}
		\Theta\left(\sqrt{K/M}\right) &: &\quad 0 < \gamma <1, \\
		\Theta\left(\sqrt{K/M \log K}\right) &: &\quad  \gamma=1, \\
		\Theta\left(K^{1-\gamma/2}/\sqrt{M}\right) &: &\quad  1< \gamma <2, \\ 
		\Theta\left(\log K /\sqrt{M}\right) &: &\quad \gamma=2, \\
		\Theta\left(1/\sqrt{M}\right) &: &\quad \gamma>2.
		\end{array}
		\right.	
		\end{equation} 
\end{thm}

\begin{proof}[Proof of Theorem~\ref{thm:Strategy_I_CommCost}]	
	Assume a request from an arbitrary node $u$ for file $W_j$. The probability that this file is cached at another node $v$ is $q_j :=1-(1-p_j)^M$. Cache content placement at different nodes is independent. Thus, the  number of nodes which should be probed is a geometric random variable with success probability $q_j$. This results in the average $1/q_j$ trials that leads to expected distance of 
	\begin{equation}
	\Theta\left(\frac{1}{\sqrt{q_j}}\right)=\Theta\left(\frac{1}{\sqrt{1-(1-p_j)^M}}\right).
	\end{equation}
	When averaged over different files we will have
	\begin{equation}
	C=\sum_{j=1}^K {p_j\Theta\left(\frac{1}{\sqrt{1-(1-p_j)^M}}\right)}.
	\end{equation}
	\begin{itemize}
		\item For Uniform distribution we have $p_j=1/K$ and then
		\begin{equation}
		C=\Theta(\sqrt{K/M}).
		\end{equation}
		\item For Zipf distribution with $M=\Theta(1)$ we have
		\begin{align}\label{final}
		C &= \sum_{j=1}^K {p_j\Theta\left(\frac{1}{\sqrt{1-(1-p_j)^M}}\right)} \nonumber \\
		&= \sum_{j=1}^K p_j \Theta \left( \frac{1}{\sqrt{p_j M}} \right) \nonumber \\ 
		&= \Theta\left(      		\frac{\sum_{j=1}^Kj^{-\gamma/2}}{\left(M\sum_{j=1}^K j^{-\gamma}\right)^{1/2}}  \right). \nonumber \\
		\end{align}
	\end{itemize}
	Define $\Lambda(\gamma) :=\sum_{j=1}^Kj^{-\gamma}$, for every $\gamma$.		
	On the other hand  it is known that  for every $\gamma>0$ (\eg, see \cite{JiTLC15})

	\begin{equation}\label{zipfestimate}
	\Lambda(\gamma)=\left\{
	\begin{array}{ll}
	\Theta\left(K^{1-\gamma}\right),  & \quad 0 < \gamma <1, \\
	\Theta\left(\log K\right), &\quad \gamma=1, \\
	\Theta(1), &\quad \gamma>1.
	\end{array}
	\right.	
	\end{equation}

	Now inserting the above equations  into (\ref{final}) completes the proof. 
\end{proof}

Theorem \ref{thm:Strategy_I_CommCost}  shows how non-uniform file popularity reduces communication cost. The skew in file popularity is determined by the parameter $\gamma$ which will affect the communication cost. For example, for $\gamma<1$ communication cost is similar to the Uniform distribution, while for $\gamma > 2$, it becomes independent of $K$. 

Since in Strategy I we have assigned each request to the nearest replica, Theorem~\ref{thm:Strategy_I_CommCost} characterizes  the minimum communication cost one can achieve. 
However, Theorems~\ref{thm:imblance}~and~\ref{thm:Strategy_I_Large} show a logarithmic growth for the maximum load as a function of network size $n$. 
This imbalance in the network load is because in Strategy I each request assignment does not consider the current load of servers. A natural question is whether, at each request allocation, one can use a very limited information of servers' current load in order to reduce the maximum load. Also one can ask how does this affect the communication cost.
We address these questions in the following section.

\section{Proximity-Aware Two Choices Strategy}\label{sec:Proximity-Aware_Two_Choices_Strategy}
Strategy I introduced in the previous section will result in the minimum communication cost, while, the maximum load for that strategy is of order $\Omega\left(\log n/\log\log n\right)$. In this section we investigate an strategy which will result in an exponential decrease in the maximum load, \ie, reduces maximum load to $\Theta\left(\log \log n \right)$, formally defined as follows.

\begin{definition}[Proximity-Aware Two Choices Strategy]
For each request born at an arbitrary node $u$ consider two uniformly at random chosen nodes from $B_r(u)$, that have cached the requested file. Then, the request is assigned to the node with lesser load. Ties are broken randomly.
\end{definition}

For the sake of illustration, first, we consider some examples in the following.

\begin{example}[$M=K$ and $r=\infty$\footnote{It should be noted that $r\ge \sqrt{n}$ (including $r=\infty$) is equivalent to $r= \sqrt{n}$. Thus in this paper we use $r=\sqrt{n}$ and $r=\infty$ alternatively.}]\label{ex:PO2C_M=K_r=infty}
In this example each node can store all the library and there is no constraint on proximity. As mentioned in \S~\ref{sec:Introduction}, the number of files that should be handled by each node (\ie,  $D_i$ for $i=1,\dots,n$) will be a $\mathrm{Po}(1)$ random variable.  In this case, according to Strategy~II, two random nodes are chosen from all network nodes and the request is assigned to the node with lesser load.

Therefore, in terms of maximum load, this  problem is reduced to the standard power of two choices  model in the balanced allocations literature \cite{ABKU99}. In this model there are $n$ bins and $n$ sequential balls which are randomly allocated to bins. In every round each ball picks two random bins uniformly, and it is then allocated to the bin with lesser load \cite{ABKU99}. Then it is shown that the maximum load of network is $L=\max_i T_i=\log \log n(1+o(1))$ w.h.p. \cite{ABKU99}, which is an exponential improvement compared to Strategy I. 
\end{example}
  
However, in contrast to Example~\ref{ex:PO2C_M=K_r=infty}, in cache networks usually each node can store only a subset of files, and this makes the problem different from the standard balls and bins model, considered in \cite{ABKU99}. Here, due to the memory constraint at each node, the choices are much more limited than the $M=K$ case. In other words here we have the case of \emph{related choices}.
In the related choices scenario, the event of choosing the second choice is correlated with the first choice; this correlation may annihilate the effect of power of two choices as demonstrated in Example~\ref{ex:PO2C_K=n_M=cte_r=infty}.

\begin{example}[$K=n$, $M=\Theta(1)$, and $r=\infty$]\label{ex:PO2C_K=n_M=cte_r=infty}
In this regime,  there is a subset of the library, say $S$, with $|S|=\Theta(n)$, whose files are cached inside the network. On the other hand, each file type is requested $\mathrm{Po}(1)$ times and hence w.h.p. there will be a file in $S$ which is requested $\Theta(\log n/\log\log n)$ times (e.g., see \cite{Devroye85}). Since each file in $S$ is replicated at most $M$ times, requests for the file are distributed among at most $M$ nodes and thus
  the maximum load of the corresponding nodes will be at least $\Theta(\log n/\log\log n)/M$. Hence, due to memory limitation we cannot benefit from the power of two choices.
\end{example}

Although Example~\ref{ex:PO2C_K=n_M=cte_r=infty} shows that memory limitation can annihilate the power of two choices but this is not always the case. Example~\ref{ex:prop_M1} shows even for $M=1$ for some scenarios we can achieve $L=O(\log\log n)$.

\begin{example}[$K=n^{1-\epsilon}$ for every constant $0<\epsilon<1$, $M=1$, and $r=\infty$] \label{ex:prop_M1}
For any popularity distribution $\mc{P}$ where $\sum_{j=1}^{K}{(p_jn)^{-c}}=o(1)$, Strategy~II achieves maximum load $L=O(\log \log n)$ w.h.p. Also, notice that Uniform and Zipf distributions   satisfy this requirement,
whenever $\epsilon\in \left( \frac{\gamma-1}{\gamma}, 1 \right)$ for $\gamma\ge 1$, where $\gamma$ is the Zipf parameter.

Roughly speaking, when $M=1$, we may partition the servers based on their cached file and hence we have $K$ ``disjoint'' subsets of servers. Similarly there are $K$ request types where each request should be addressed by the corresponding subset of servers. Thus, here we have $K$ disjoint Balls and Bins sub-problems, and the sub-problem with maximum load determines the maximum load of the original setup. 
The reason that here, in contrast to Example~\ref{ex:PO2C_K=n_M=cte_r=infty}, we can benefit from power of two choices is the assumption of $K\ll n$.

\end{example}

\begin{proof}[Proof of Example~\ref{ex:prop_M1}]
	It is easy to see that for $M=1$, the number of caching servers with a specific file, say $W_j$ denoted by $S_j$, is distributed as a $\mathrm{Bin}(n, p_j)$. Thus applying a  Chernoff bound for $S_j$ (\eg, see Appendix~\ref{app:bounds}) implies that 
	\[
	\Pr{|S_j-\Ex{S_j}|\geq \Ex{S_j}/2}\le 2\exp({-p_j n/12}).
	\]
	Moreover, let $R_j$ denote the number of requests for file $W_j$, which is the sum of $n$ i.i.d. $\mathrm{Bin}(n, p_j)$ random variables. Again applying a Chernoff bound (\eg, see Appendix~\ref{app:bounds}) for Poisson random variables yields that
	\[
	\Pr{|R_j-\Ex{R_j}|\geq \Ex{R_j}/2}\le 2\exp({-p_j n/12}).
	\]
	Notice that $\Ex{S_j}=\Ex{R_j}=np_j$. 
	Suppose that  $\mathcal{A}_j$ denotes the event that $|S_j-\Ex{S_j}|\leq \Ex{S_j}/2$ and $|R_j-\Ex{R_j}|\leq \Ex{R_j}/2$. Then we have that $\Pr{\mathcal{A}_j}\ge 1-4\exp(-p_jn/12)$. Also define $\mathcal{E}_j$  to be  the event that the two-choice model with $S_j$ bins (caching servers) and $R_j$ balls (requests) achieves maximum load $R_j/S_j+\Theta(\log\log S_j)$. It is shown that this event happens with probability  $1-O(1/S_j^c)$, for every constant $c$ (e.g., see \cite{ABKU99}).
	So we have that 
	\begin{align*}
	\Pr{\mathcal{E}_j}&=\Pr{\mathcal{E}_j| \mathcal{A}_j}\Pr{\mathcal{A}_j}+
	\Pr{\mathcal{E}_j| \neg\mathcal{A}_j}\Pr{\neg\mathcal{A}_j}\\
	&>(1-2(p_j n)^{-c})(1-4\exp(-p_jn/12))\\
	&+ (1-2(p_j n)^{-c})(4\exp(-p_jn/12))\\
	&\ge 1-8(p_j n)^{-c}.
	\end{align*}
	Since we have $K$ disjoint subsystems, the union bound over all subsystems shows that the two choice model does achieve  the desired maximum load with probability 
	$1-8\sum_{j=1}^K(np_j)^{-c}=1-o(1)$ which concludes the proof due to example's assumption on popularity profile. 
	
	Now we show that the Uniform and Zipf distributions satisfy the example's assumption. When $\mathcal{P}$ is the Uniform distribution over $K$ files, $\forall j,\ p_j\cdot n=n^{\epsilon}$. Now by setting $c=3/\epsilon$, we have that 
	\[
	\sum_{j=1}^K(np_j)^{-c}=K({1}/{n^\epsilon})^c=K/n^3=o(1/n^2).
	\]
	
	Also, for Zipf distribution we have
	\[
	p_j=\frac{j^{-\gamma}}{\sum_{j=1}^K j^{-\gamma}}=\frac{j^{-\gamma}}{\Lambda(\gamma)}.
	\]
	Depending on $\gamma$,  we consider two cases:
	\begin{itemize}
		\item $\gamma\ge 1$: For every $c>1$ we have
		\[
		\left(\frac{\Lambda(\gamma)}{n}\right)^c\Lambda(\gamma c)=\Theta\left(\frac{\log^c K}{n^c}\right)\Lambda(\gamma c)=o(1),
		\]
		where we used $K<n$ and \eqref{zipfestimate}.
		\item $0< \beta<1$: By setting 	 $c=2/ \gamma$ and using the fact that $K<n$, we have
		\begin{align*}
		\left(\frac{\Lambda(\gamma)}{n}\right)^c
		\Lambda(\gamma c)=&\Theta\left( \frac{K^{(1-\gamma)c}}{n^c}\right)
		\le {n^{(1-\gamma)c-c}}\Lambda(\gamma c)\\
		=&n^{-\gamma c}\Lambda(\gamma c)=o(1),
		\end{align*}
		where we applied \eqref{zipfestimate}.
	\end{itemize}
\end{proof}

Above examples bring to attention the following question.

\begin{question}\label{que:MemoryLimitRegimes}
In view of the memory limitation at each server in cache networks, what are the regimes (in terms of problem parameters) one can benefit from the power of two choices to balance out the load?
\end{question}

Addressing Question~\ref{que:MemoryLimitRegimes}, for the general $M>1$ case, is more challenging than Example~\ref{ex:prop_M1} and needs a completely different approach.
The simplicity of case $M=1$ is that there is no interaction between $K$ Balls and Bins sub-problems. On the other hand, consider $M>1$. If a request, say $W_j$, should be allocated to a server then the load of two candidate bins that have cached $W_j$ should be compared. However, load of other file types will also be accounted for in this comparison. So there is flow of load information between different sub-problems which makes them entangled. 

In all above examples, we have not considered the proximity constraint, \ie, $r=\infty$, yet. This results in a fairly high communication cost $C=\Theta\left(\sqrt{n}\right)$. However, in general since parameter $r$ controls the communication cost, it can be chosen to be much less than the network diameter, \ie, $\Theta(\sqrt{n})$. This proximity awareness introduces another source of correlation (other than the memory limitation) between the two choices. Thus, considering the proximity constraint may annihilate the power of two choices even in large memory cases as demonstrated in the following example.

\begin{example}[$M=K$ and $r=1$]
In this example, when a request arrives at a server, the server chooses two random choices among itself and its neighbours. Then the request is allocated to the one with lesser load. Since there exists a server at which $\max_i D_i=\Theta(\log n/\log\log n)$ requests arrive, maximum load of network (\ie, $L=\max_i T_i$) will be at least $\Theta(\log n/\log\log n)/5$.
\end{example}

Thus, similar to Question~\ref{que:MemoryLimitRegimes} regarding the memory limitation effect, one can pose the following question regarding proximity principle.

\begin{question}\label{que:ProximityConstraint}
In view of the proximity constraint of Scheme~II, what are the regimes (in terms of problem parameters) one can benefit from the power of two choices to balance out the load?
\end{question}

In order to completely analyze load balancing performance of Scheme~II, one should consider both sources of correlation simultaneously (which is not the case in above examples). To this end, in the following, we investigate two memory regimes, namely $M=K$ and $M=n^{\alpha}$ for some $0<\alpha<1/2$.

Our main result for $M=n^{\alpha}$ is presented in the following theorem.

\begin{thm}\label{main:twochoice}
Suppose that $0< \alpha, \beta<1/2$ be two constants and let $K=n$, $M=n^{\alpha}$, and $r=n^{\beta}$. Then if 
\[
\alpha+2\beta\ge 1+2\log\log n /\log n,
\]
under the Uniform popularity distribution,  Strategy II achieves maximum load $L=\Theta(\log \log n)$ and communication cost $C=\Theta(r)$ w.h.p.
\end{thm}

\begin{remark}
To have a more accessible proof, in Theorem~\ref{main:twochoice}, we have assumed that $K=n$.
Note that the proof techniques can also be extended to the case where $K=O(n)$.
\end{remark}

In order to prove the theorem, let us first present an interesting result that was shown in \cite{KP06} as follows.
\begin{thm}[\cite{KP06}]\label{thm:related}
	Given an almost $\Delta$-regular graph\footnote{A graph is said to be almost $\Delta$-regular, if each vertex has degree $\Theta(\Delta)$.} $G$ with $e(G)$ edges and $n$ nodes
	representing $n$ bins, if $n$ balls are thrown into the bins
	by choosing a random edge with probability at most $O(1/e(G))$   and placing into the smaller
	of the two bins connected by the edge, then the maximum
	load is $\Theta(\log \log n) + O\left( \frac{\log n}{\log(\Delta/\log^4 n )} \right) + O(1)$ w.h.p.
\end{thm}
\begin{remark}
Note that in the original theorem presented in \cite{KP06}, it is assumed that each edge is chosen uniformly among all edges of graph $G$. However, here we slightly generalize the result so that each edge is chosen with probability at most $O(1/e(G))$. The proof follows the original proof's idea with some  modifications in computation parts, where due to lack of space we omit.
\end{remark}
In order to apply Theorem~\ref{thm:related}, we first need to define a new graph $H$ as follows.
\begin{definition}[Configuration Graph]
For  the given parameter $r$, configuration graph $H$ is defined as a graph	whose nodes represent  the servers and two nodes, say $u$ and $v$, are connected if and only if  $u$ and $v$ have cached a common file and $d(u, v)\le 2r$ in the torus.
\end{definition}

For every two servers $u$ and $v$,
let $T(u, v)$ be the set of distinct files that have been cached in both nodes $u$ and $v$. Also denote $|T(u, v)|$ by $t(u, v)$. Define $t(u)$ to be the number of distinct  cached files  in $u$. Now, let us define \emph{goodness} of a placement strategy as follows.
\begin{definition}[Goodness Property]
For every positive constant $\delta\in[0, 1]$ and $\mu=O(1)$, we say the  file placement strategy is \emph{$(\delta, \mu)$-good}, if for every $u$ and $v$, $t(u)\ge \delta M$ and $t(u, v)<\mu$.
\end{definition}

\begin{lemma}\label{lem:GoodnessProperty}
The proportional cache placement strategy introduced in \S\ref{sec:Notation_ProblemSetting}, is $(\delta, \mu)$-good w.h.p.
\end{lemma}
\begin{proof}
	Clearly, every set of cached files in every node (with replacement) can be one-to-one mapped to a non-negative integral solution of equation $\sum_{i=1}^Kx_i=M$, where each $x_i$ expresses the number of times  that file $i$  has been cached in the node. A combinatorial argument shows that, the equation has ${K+M-1 \choose M}$ non-negative integer solutions. So for each $1\le s\le  M$, we have
	\begin{equation}\label{eq:Prob_c(u)=s}
	\Pr{t(u)=s}=\frac{{ K \choose s}{M-1 \choose M-s}}{{ K+M-1 \choose M}},
	\end{equation}
	where we first fixed a set of file indexes   of size $s$, say $I=\{i_1, i_2,\ldots, i_s\}$, and then count the number of integral solutions to the equation $\sum_{i\in I}x_i=M-s$.

In order to bound \eqref{eq:Prob_c(u)=s}, we note that for every $1\le a\le b,$ $(b/a)^a \le  {b \choose a} \le b^a$ and also ${b \choose a}\le 2^b$.  Recall that we assumed  $K=n$ and $M=n^\alpha$, $0<\alpha<1/2$. Hence for every $1\leq s \leq \delta M$, we have
	\begin{align*}
	\Pr{t(u)=s}&\le \frac{K^s2^{M}}{{K\choose M}}\le 
	\frac{K ^s2^{M}}{(K/M)^M}=
	{(2M)^M}{K^{s-M}}\\
	&
	\leq  (2n^\alpha n^{\delta-1})^M.
	\end{align*}
 Thus, by choosing $\delta=(1-\alpha)/3$, for every \mbox{$1\le s\le \delta M $}, we have
	\begin{align*}
	\Pr{t(u)=s}&\le (2n^{\alpha+\delta-1})^M=
	(2n^{2\alpha/3-2/3})^M\\
	& \le (2n^{-1/3})^M=n^{-\omega(1)},
	\end{align*}
	where the last equality follows from $M=n^{\alpha}=\omega(1)$.
	Now the union bound over all $1\le s\le \delta M$ and $n$ nodes yields 
	\begin{align}\label{prob:size}
	\Pr{\exists u\in V: t(u) \le \delta M}=n^{-\omega(1)}.
	\end{align}
	By a similar argument, for each $1\le t\le M$ and every $u$ and $v$, we have
	\begin{align*}
	\Pr{t(u, v)\geq t}={K\choose t}\left(\frac{{K+M-t-1 \choose M-t}}{{K+M-1 \choose M}}\right)^2.
	\end{align*}
	Thus,  for any constant $\mu \ge 5/(1-2\alpha)$, we can write
	\begin{align*}
	&\Pr{t(u, v)\ge \mu}\\
	&\le K^\mu \left(\frac{(K+M-\mu-1)!M!}{(K+M-1)!(M-\mu)!}\right)^2\\
	&\le K^\mu \left(\frac{M^\mu}{K^\mu}\right)^2\le \frac{M^{2\mu}}{K^\mu}= n^{(2\alpha-1)\mu}=O(1/n^5).
	\end{align*}
 By applying the union bound over all pairs of servers, for every $u$ and $v$ we have 
	\begin{align}\label{prob:com}
	\Pr{t(u, v)\ge  \mu}=O(1/n^3).
	\end{align}
	Hence, $t(u, v) < \mu$ w.h.p.
Putting (\ref{prob:size}) and (\ref{prob:com}) together  concludes the proof.
\end{proof}

The following lemma presents some useful properties of $H$ and Strategy II.
\begin{lemma}\label{lem:Gp}
	Conditioning on ``goodness'' of the file placement and assuming $K=n$, $M=n^\alpha$ and $r=n^{\beta}$ with $\alpha+2\beta\ge 1+2\log\log n/\log n$, we have
	\begin{itemize}
		\item[(a)] W.h.p. $H$ is almost $\Delta$-regular with $\Delta=\Theta\left(\frac{M^2r^2}{K}\right)$.
		\item[(b)] For each request, Strategy II samples an edge of $H$ (two servers) with probability $O(1/e(H))$.
	\end{itemize}
\end{lemma}
\begin{proof}
	Consider an arbitrary node $u$ with $s$ distinct files. Then by definition of $H$, for every node $v$  we have 
	\begin{align*}
	p_s:=\Pr{t(u, v)\ge 1| t(u)=s }&=1-\left(\frac{K-s}{K}\right)^M\\
	&=\frac{sM}{K}(1+o(1)),
	\end{align*}
	where $1\le s\le M$.
	On the other hand $u$ and $v$ are connected in $H$, if in addition we have $d_G(u, v)\leq 2r$.
	Therefore for every given node $u$ with $s$ distinct cached files, $d(u)$ in $H$ (degree of $u$ in $H$) has a binomial distribution $\mathrm{Bin}(b_{2r}(u), p_s)$, where  $b_{2r}(u)=|B_{2r}(u)|$. Hence applying a Chernoff bound implies that with probability $1-n^{-\omega(1)}$, we have 
	\[
	d(u)=\frac{sMb_{2r}(u)}{K}(1+o(1)).
	\] 

	Conditioning on the goodness of file placement, $s=t(u)=\Theta(M)$.
	Also  by symmetry of torus, we have $b_{2r}(u)=\Theta(r^2)$, for every $u$. So, with high probability for every $u$, we have
	\[
	d(u)=\Theta\left({M^2r^2}/{K}\right),
	\]
	where this concludes the proof of part (a).

	Now it remains to show that Strategy~II picks an edge of $H$, with probability $O(1/e(H))$. First, notice that
	\begin{align}\label{eq:edgesize} e(H)=\Theta\left({nM^2r^2}/{K}\right)=\Theta(M^2r^2),
	\end{align}
    as $K=n$. 
	Then recall that each file is cached in every node with probability 
	$p=1-(1-1/K)^M=M/K(1+o(1))$, independently. 
	For any given node $u$ and file $W_j$, let $F_j(u)$ be the number of nodes at distance at most $r$ that have cached file $W_j$. Then $F_j(u)$ has a binomial distribution $\mathrm{Bin}(b_r(u), p)$, where $b_r(u)=|B_r(u)|$. So 
	\[
	\Ex{F_j(u)}=b_r(u)\cdot p=\Theta(r^2M/K),
	\] 
	where $b_r(u)=\Theta(r^2)$ for every $u$.
	Since \mbox{$\alpha+2\beta\ge 1+2\log\log n/\log n$}, we have $\Ex{F_j(u)}=\omega(\log n)$, for every  $u$ and $j$.
	Now, applying a Chernoff bound for $F_j(u)$ implies that with probability $1-n^{-\omega(1)}$, $F_j(u)$ concentrates around its mean and hence, w.h.p., we have for every $u$ and $j$
	 \[	
	F_j(u)=\Theta(r^2M/K)=\Theta(r^2M/n).
	\] 
	
Consider an edge $(u, v)\in E(H)$, with $t(u, v)=t$. Define $S_{u,v}$ to be  the set  of nodes that may pick pair $u$ and $v$ randomly in Strategy~II. It is not hard to see that $|S_{u, v}|=O(r^2)$. Now we have,	
	\begin{align}
	&\Pr{ (u, v)\in E(H) \text{ is picked by Strategy II}| t(u, v)=t} \nonumber\\
	&\quad\quad =\sum_{j\in T(u, v)}\frac{1}{K}\sum_{w\in S_{u, v}}\frac{1}{n}\frac{1}{{F_j(w)\choose 2}} \nonumber\\
	&\quad\quad =\frac{1}{n^2}
	\sum_{j\in T(u, v)}\sum_{w\in S_{u, v}}\frac{1}{{F_j(w)\choose 2}} \nonumber\\
	&\quad\quad =\frac{1}{n^2}
	\sum_{j\in T(u, v)}\sum_{w\in S_{u, v}}
	\Theta({n^2}/{r^4M^2}). \label{eq:CondPr_(u,v)_picked}
	\end{align}
Conditioned on ``goodness,'' we have  for every $(u, v) \in E(H)$, $1\le t(u, v)< \mu$. So  \eqref{eq:CondPr_(u,v)_picked} can be simplified as 
\begin{align*}
&\Pr{ (u, v)\in E(H) \text{ is picked by Strategy II}}\\
&\quad\quad\le \Theta({ \mu |S_{u, v}|}/{r^4M^2})\\
&\quad\quad =O(1/r^2M^2)=O(1/e(H)),
	\end{align*}
	where the last equality follows from \eqref{eq:edgesize}.
\end{proof}

\begin{proof}[Proof of Theorem \ref{main:twochoice}]
	Applying Lemma \ref{lem:Gp} shows that w.h.p. the configuration graph $H$ is an almost $\Delta$-regular graph where $\Delta=M^2r^2/n$. Moreover, in each step, every edge of $H$ is chosen randomly with probability $O(1/e(H))$. 
	Hence, we can apply Theorem~\ref{thm:related} and conclude that w.h.p. the maximum load is at most
	\[
\Theta(\log\log n)+ O\left( \frac{\log n}{\log(\Delta/\log^4 n )} \right) =\Theta(\log\log n)+ O(1),
	\]
where it follows because $\alpha+2\beta\ge 1+2\log\log n/\log n$ and hence $\Delta=M^2r^2/n=n^{2\alpha+2\beta-1}>n^{\alpha}$.
\end{proof}

Now let us present our next result regarding to the $M=K$ regime.

\begin{thm}\label{thm:full}
 Suppose $M=K$ and Uniform distribution $\mathcal{P}$ over the file library. Then Strategy II achieves the maximum load $L=\Theta\left(\log \log n\right)$ and communication cost $C=\Theta\left(n^{\beta}\right)$ for any  $\beta=\Omega(\log\log n/\log n)$.
 \end{thm}
 
 \begin{proof}
  Let us choose $r=n^{\beta}$, for some $\beta=\Omega(\log\log n/\log n)$. 
By the assumption $M=K$, the configuration graph $H$ (corresponding to $r$) is a graph in which two nodes $u$ and $v$ are connected if and only if $d(u,v)\le 2r$. Since our network is symmetric, for every $u$, $|B_r(u)|=\Theta(r^2)$ and hence $H$ is a  regular graph with $\Delta=\Theta(r^2)$.  
Also it is not hard to see that  Strategy II is equivalent to choosing an edge uniformly from $H$. Applying Theorem \ref{thm:related} (\cite{KP06}) to $H$ results in the maximum load of $\Theta(\log\log(n))$. In addition, choosing two random nodes in $|B_r(u)|=\Theta(r^2)$ results in communication cost of $C=\Theta(r)=\Theta\left(n^{\beta}\right)$. 

 \end{proof}
 
 The main point of Theorem \ref{thm:full} is that we should just endure $C=\Theta\left(n^{\beta}\right)$, for $\beta=\Omega(\log\log n/\log n)$, to benefit from the luxury of power of two choices, which is a very encouraging result.

\section{Simulations}\label{sec:Simulations}
In this section, we demonstrate the simulation results for two strategies introduced in the previous sections, namely, \emph{nearest replica} and \emph{proximity aware two choices} strategies. Our simulations are implemented in Python where the code is available online at \cite{OurSimulationCode_Github}. 

Figure~\ref{fig:OneChoice_SrvSzVar_MaxLoad} shows the maximum load of Strategy~I as a function of the number of servers where different curves correspond to different cache sizes. The network graph is a torus, where $100$ files with Uniform popularity are placed uniformly at random in each node. Each point is an average of $10000$ simulation runs.
This figure is in agreement with the logarithmic growth of the maximum load, asymptotically proved in \S\ref{sec:Nearest_Replica_Strategy}, even for the intermediate values of $n\approx 100,\ldots, 3000,$ which makes the results of \S\ref{sec:Nearest_Replica_Strategy} more general. 
Comparing different curves reveals the fact that in larger cache size setting, we have a more balanced network. That happens because enlarging cache sizes results in a more uniform Voronoi tessellation, \ie, having cells with smaller variation in size.

Furthermore, Figure~\ref{fig:OneChoice_CacheSzVar_AvgCost} shows the communication cost of Strategy~I as a function of cache size where different curves correspond to different library sizes. Here, the network graph is a torus of size $2025$ and each point is an average of $10000$ simulation runs. This figure is in agreement with the result of Theorem~\ref{thm:Strategy_I_CommCost}.

\begin{figure}
\begin{center}
\includegraphics[width=0.48\textwidth]{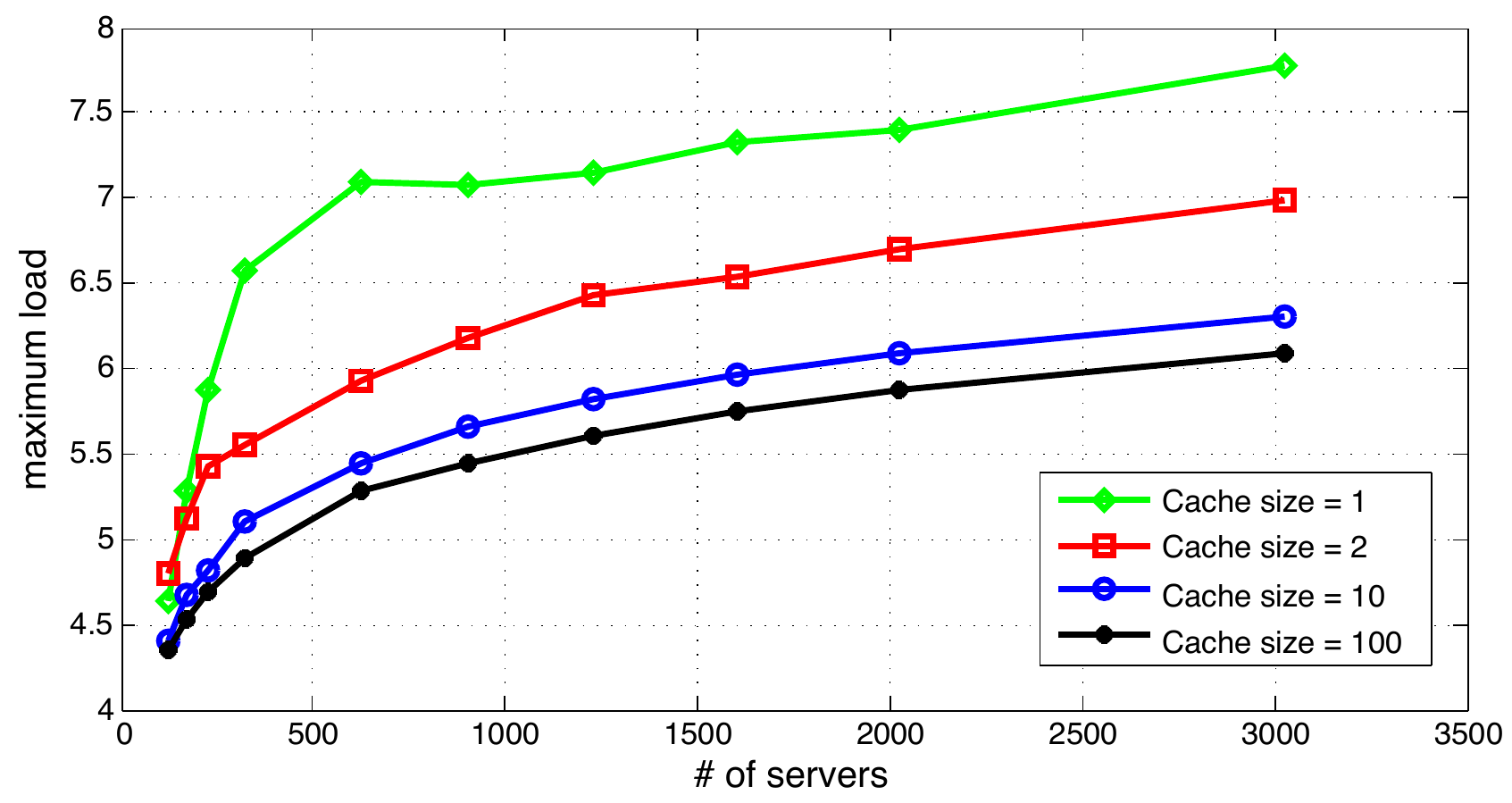}
\end{center}
\caption{The maximum load versus number of servers for Strategy~I. Here, the network topology is a torus, the file popularity is Uniform, and we have $K=100$.} 
\label{fig:OneChoice_SrvSzVar_MaxLoad}
\end{figure}

\begin{figure}
\begin{center}
\includegraphics[width=0.48\textwidth]{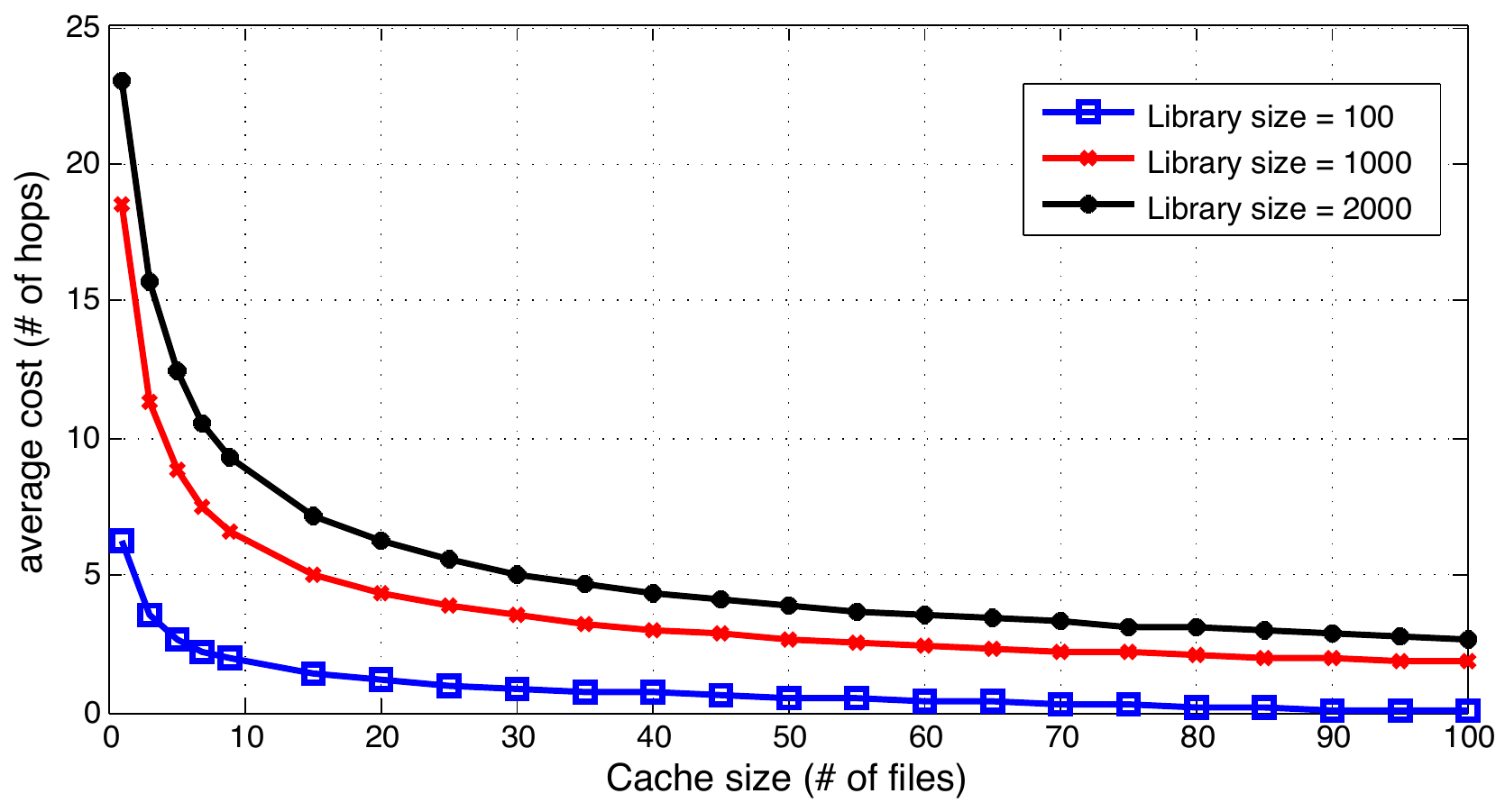}
\end{center}
\caption{The communication cost versus cache size for Strategy~I. Here, the network topology is a torus of size $2025$ and the file popularity is Uniform.}
\label{fig:OneChoice_CacheSzVar_AvgCost}
\end{figure}

\begin{figure}
\begin{center}
\includegraphics[width=0.48\textwidth]{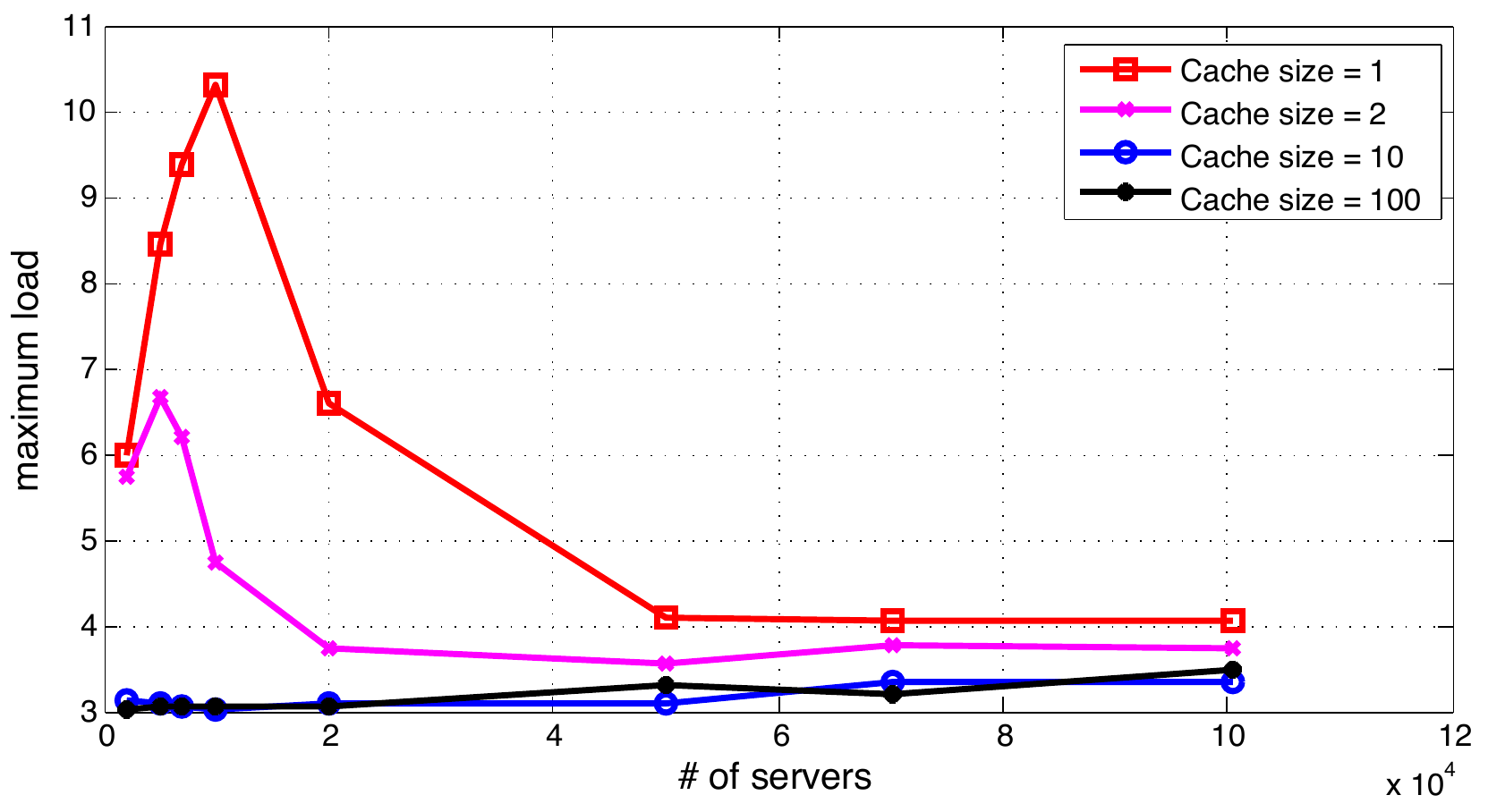}
\end{center}
\caption{The maximum load versus number of servers for Strategy~II.
Here, the network topology is torus, the file popularity is Uniform, and the library size is $K=2000$. Moreover, we assume $r=\infty$.}
\label{fig:TwoChoice_SrvSzVar_MaxLoad}
\end{figure}

In order to simulate Strategy~II, first we set $r=\infty$ to study the effect of cache size on the maximum load and communication cost, and then consider the effect of limited $r$ on the performance of the system. Figure~\ref{fig:TwoChoice_SrvSzVar_MaxLoad} shows the maximum load of the network versus number of servers where each curve demonstrates a different cache size. The network graph is a torus, where $2000$ files with Uniform popularity are placed uniformly at random in each node. Each point is an average of $800$ simulation runs.
In each curve, since cache size and number of files are fixed, increasing the number of servers translates to increasing each file replication. 
Figure~\ref{fig:TwoChoice_SrvSzVar_MaxLoad} demonstrates the system performance for large system sizes, \ie, $n\approx 10^4,\ldots,10^5$. However, to get a better understanding of network behavior, let us compare the load balancing performance of Strategies I and II in Figure~\ref{fig:Lattice-SrvSzVar-1ch_vs_2ch} where the file library size is $K=100$ and $n\approx 10^3$.

In Figure~\ref{fig:Lattice-SrvSzVar-1ch_vs_2ch}, when the file replication is low, due to high correlation between the two choices of Strategy~II, power of two choices is not expected. This is reflected in Figure~\ref{fig:Lattice-SrvSzVar-1ch_vs_2ch}; for example in the curve corresponding to $M=1$ for $n \le 400$ we have a fast growth in maximum load which mimics the load balancing performance of Strategy~I. We see the same trend for the curve corresponding to $M=2$ for $n\le 200$.
However, assuming $M\ge 2$, for $n>1000$, since there is enough file replication in the network, the load balancing performance is greatly improved due to the power of two choices. This is in accordance with the lessons learned from \S\ref{sec:Proximity-Aware_Two_Choices_Strategy}. Also, in between, we have a transition region where a mixed behavior is observed. 
Observations made above from Figure~\ref{fig:Lattice-SrvSzVar-1ch_vs_2ch} have an important practical implication. Since employing Strategy~II is only beneficial in networks with high file replication, for other situations with limited cache size, the less sophisticated Strategy~I is a more proper choice.

\begin{figure}
\begin{center}
\includegraphics[width=0.48\textwidth]{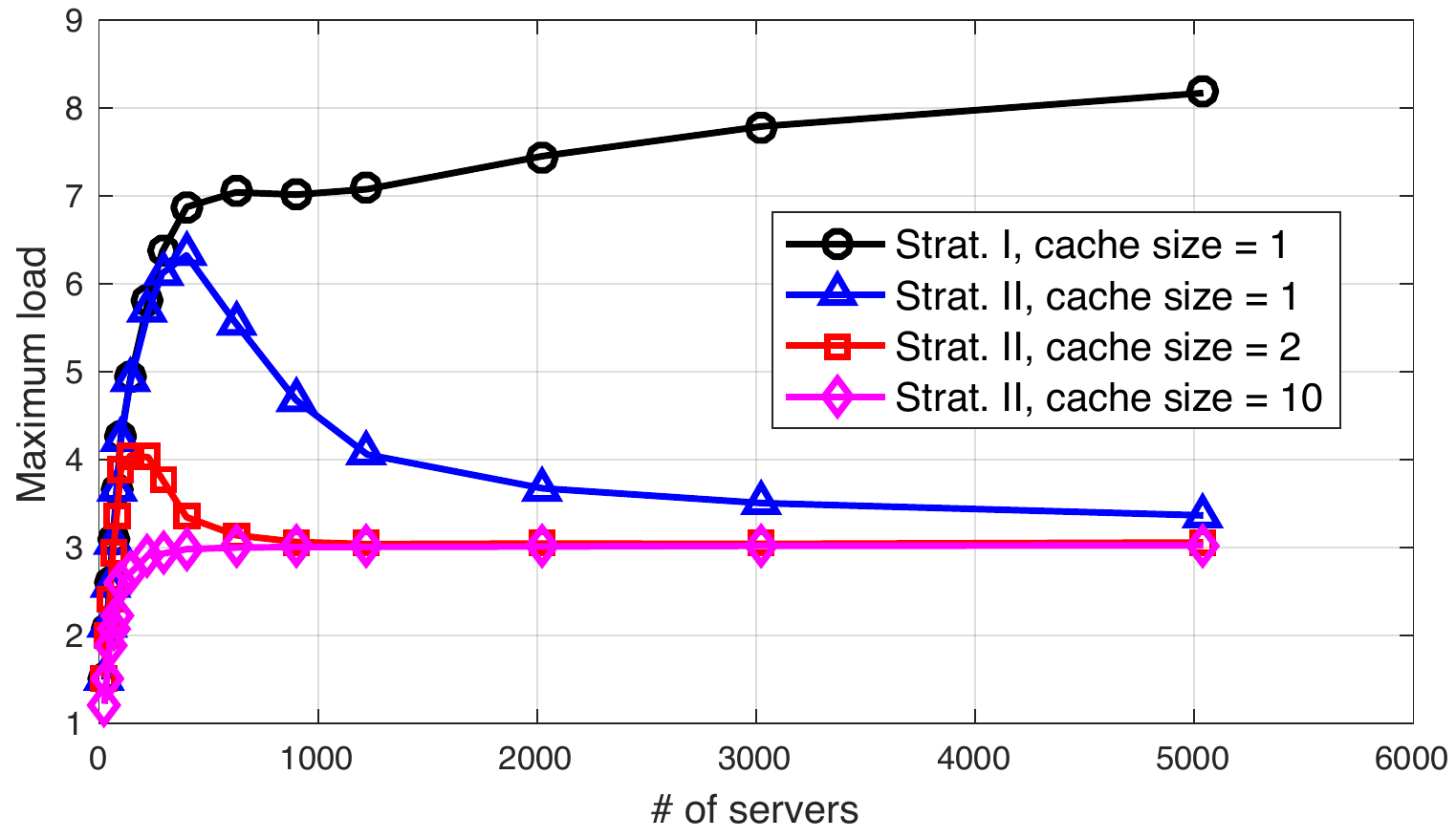}
\end{center}
\caption{The maximum load versus number of servers for Strategies~I and II. 
Here, the network topology is torus, the file popularity is Uniform, and the library size is $K=100$. Moreover, we assume $r=\infty$.}
\label{fig:Lattice-SrvSzVar-1ch_vs_2ch}
\end{figure}

Figure~\ref{fig:TwoChoice_SrvSzVar_AvgCost} draws the communication cost versus number of servers for various cache sizes for similar setting used in Figure~\ref{fig:TwoChoice_SrvSzVar_MaxLoad}. Since in this figure there is no constraint on the proximity, the communication cost growth is of order $\Theta(\sqrt{n})$.

\begin{figure}
\begin{center}
\includegraphics[width=0.48\textwidth]{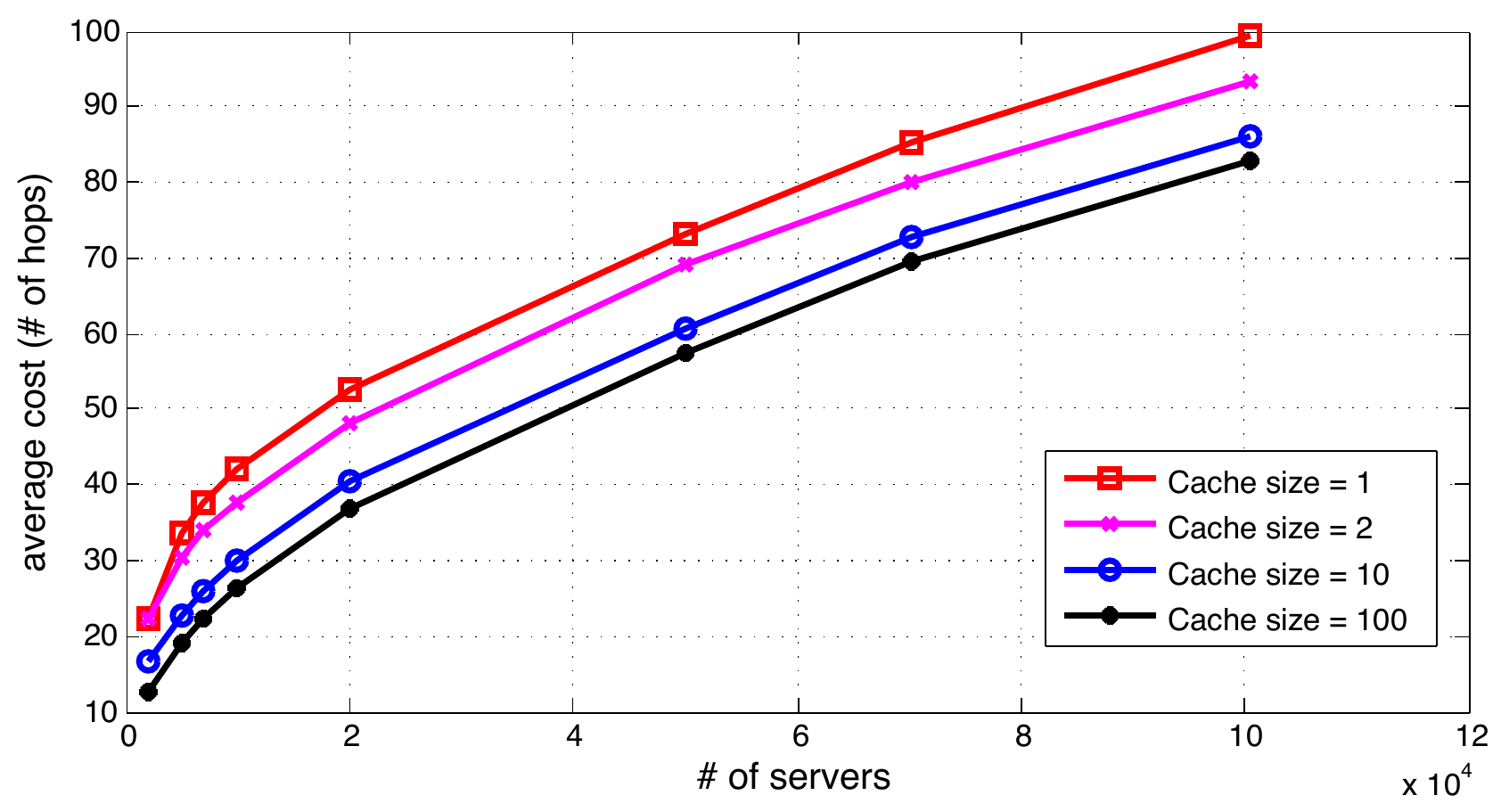}
\end{center}
\caption{The communication cost versus number of servers for Strategy~II. Here, the settings are similar to that of Figure~\ref{fig:TwoChoice_SrvSzVar_MaxLoad}.}
\label{fig:TwoChoice_SrvSzVar_AvgCost}
\end{figure}

In simulations presented so far, we only considered the case $r=\infty$. In order to investigate the effect of parameter $r$ on the performance of the system, in Figure~\ref{fig:Tradeoff_srvn=2025_fn=500_itr=5000}, we have simulated network operation for different values of $r$. This results in a trade-off between the maximum load and communication cost, as shown in Figure~\ref{fig:Tradeoff_srvn=2025_fn=500_itr=5000}.  Here we consider a torus with $2025$ servers, where $500$ files with Uniform popularity are placed uniformly at random in each node. Each point is an average of $5000$ simulation runs.

In this figure, like before (\ie, Figure~\ref{fig:Lattice-SrvSzVar-1ch_vs_2ch}), we observe two performance regimes based on the file replication in the network. In high memory regime, \eg, for curves corresponding to $M=50$ and $M=200$, we can achieve the power of two choices by sacrificing a negligible communication cost. On the other hand, in low memory regime, \ie, $M=1$, we cannot decrease the maximum load even at the expense of high communication cost values. For intermediate values of $M$, we clearly observe the trade-off between the maximum load and communication cost.

\begin{figure}
\begin{center}
\includegraphics[width=0.48\textwidth]{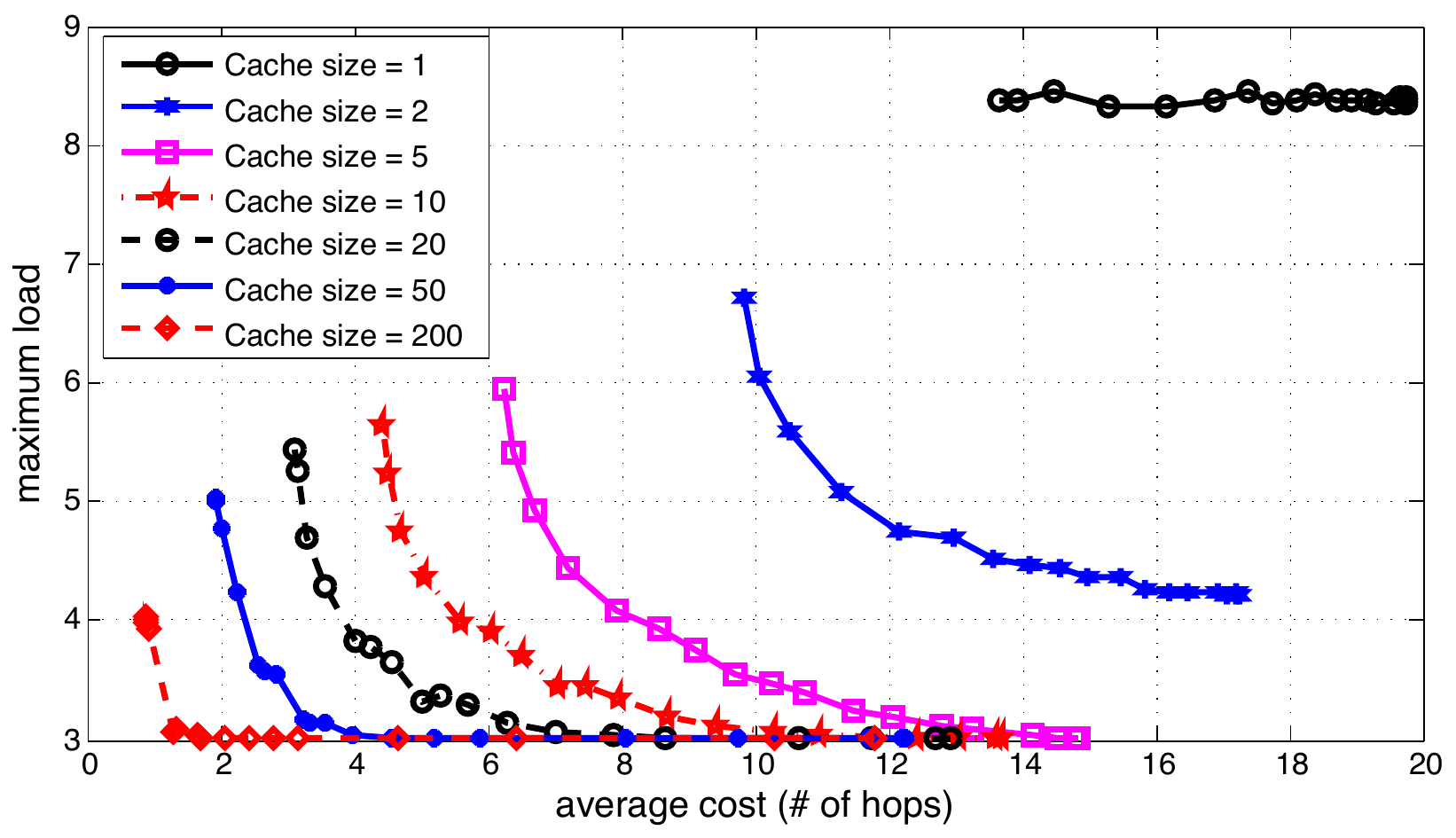}
\end{center}
\caption{The trade-off between the maximum load and communication cost for Strategy~II.
Here, the network topology is a torus of size $2025$, the file popularity is Uniform, and the library size is $K=500$.}
\label{fig:Tradeoff_srvn=2025_fn=500_itr=5000}
\end{figure}

All above simulations investigated the performance for networks with torus topology and Uniform file popularity distribution, being in agreement with our theoretical results' indications. However, one may ask how sensitive are our findings to the network topology and file popularity choices.
Thus, in the following, we examine network performance for Zipf file popularity profile and more practical network topologies, namely, random geometric graph (RGG) and power law random graph model \cite{albert2002statistical}.

Figures~\ref{fig:SrvSzVar_Lattice_Zipf_fn=64_cs=2_itr=1000},
\ref{fig:SrvSzVar_RGG_Zipf_1ch_vs_2ch_fn=64_cs=2_itr=1000}, and
\ref{fig:SrvSzVar_BarabasiAlbert_Zipf_1ch_vs_2ch_fn=64_cs=2_itr=1000}, 
show the maximum load versus the number of servers for different network topologies (namely, torus, RGG, and power law model) and Zipf distribution with parameter\footnote{Note that the Zipf distribution with $\gamma=0$ corresponds to the Uniform distribution.} $\gamma\in\{0,1,1.5\}$.

\begin{figure}
\begin{center}
\includegraphics[width=0.48\textwidth]{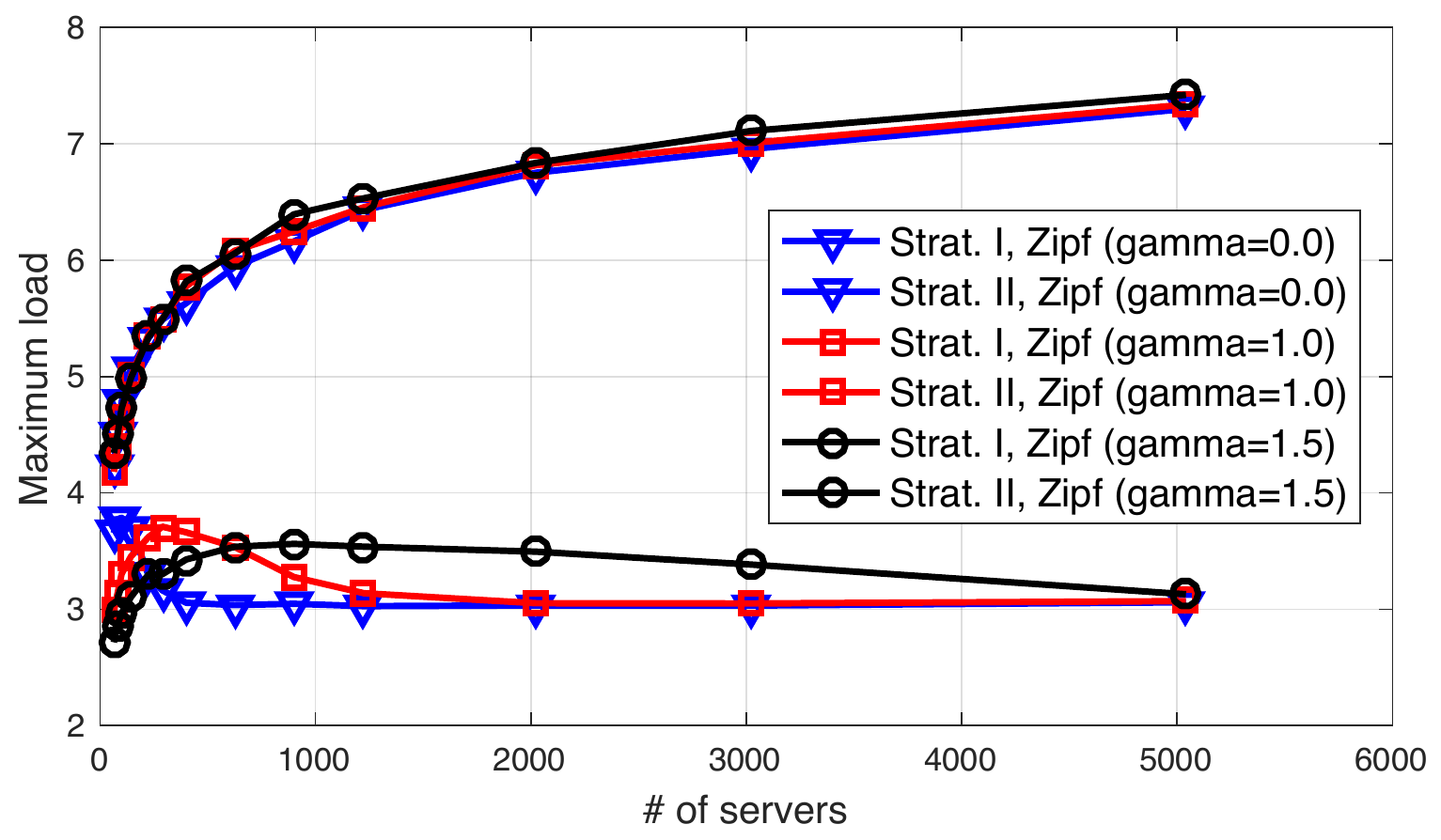}
\end{center}
\caption{The maximum load versus number of servers for Strategies~I and II. 
Here, the network topology is torus, the file popularity is Zipf, the library size is $K=64$, and the cache size is $M=2$. Moreover, we assume $r=\infty$. Each point is an average over $1000$ independent simulation runs.}
\label{fig:SrvSzVar_Lattice_Zipf_fn=64_cs=2_itr=1000}
\end{figure}

\begin{figure}
\begin{center}
\includegraphics[width=0.48\textwidth]{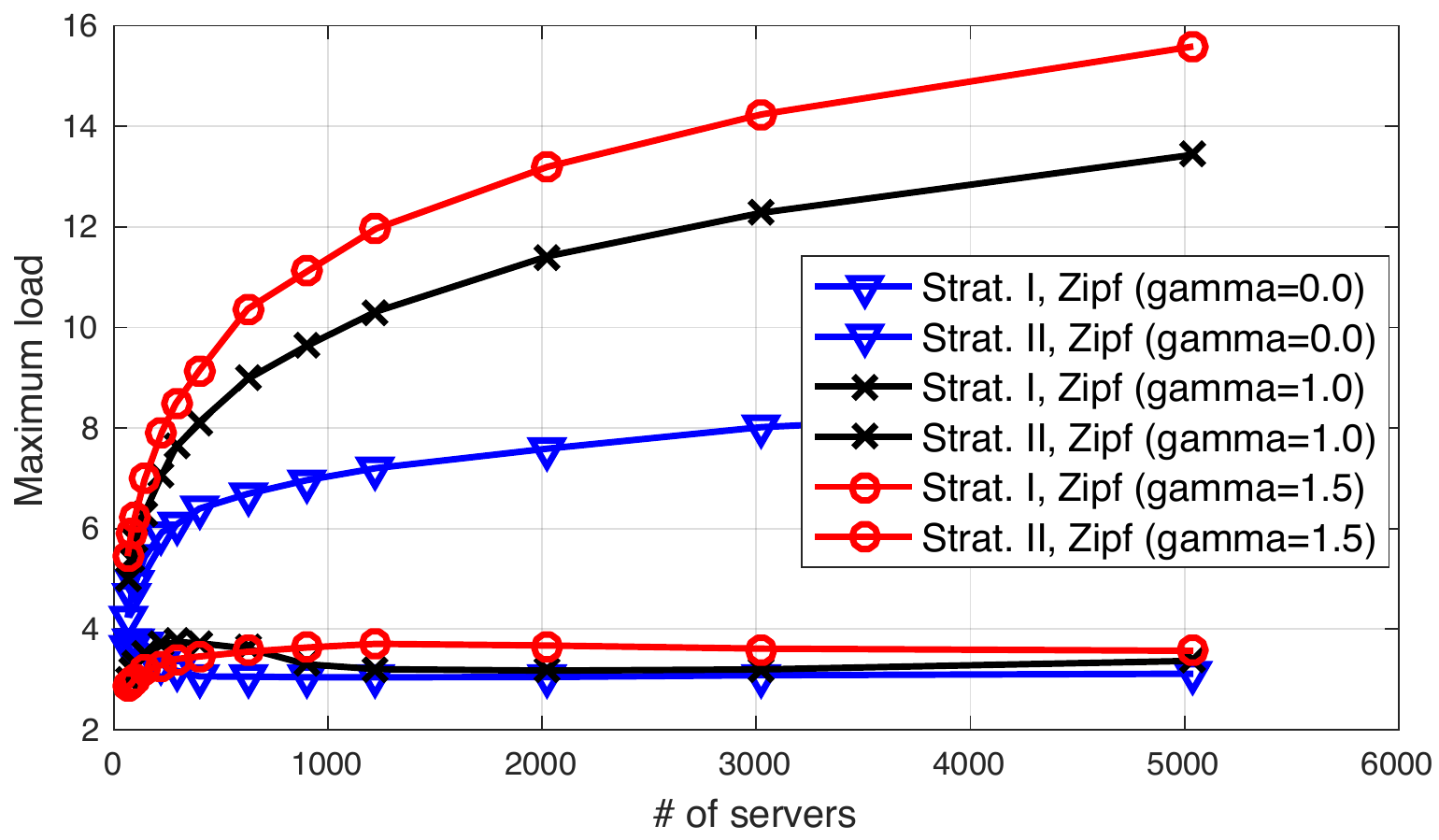}
\end{center}
\caption{The maximum load versus number of servers for Strategies~I and II. 
Here, the network topology is RGG. The remaining settings are similar to that of Figure~\ref{fig:SrvSzVar_Lattice_Zipf_fn=64_cs=2_itr=1000}.}
\label{fig:SrvSzVar_RGG_Zipf_1ch_vs_2ch_fn=64_cs=2_itr=1000}
\end{figure}

\begin{figure}
\begin{center}
\includegraphics[width=0.48\textwidth]{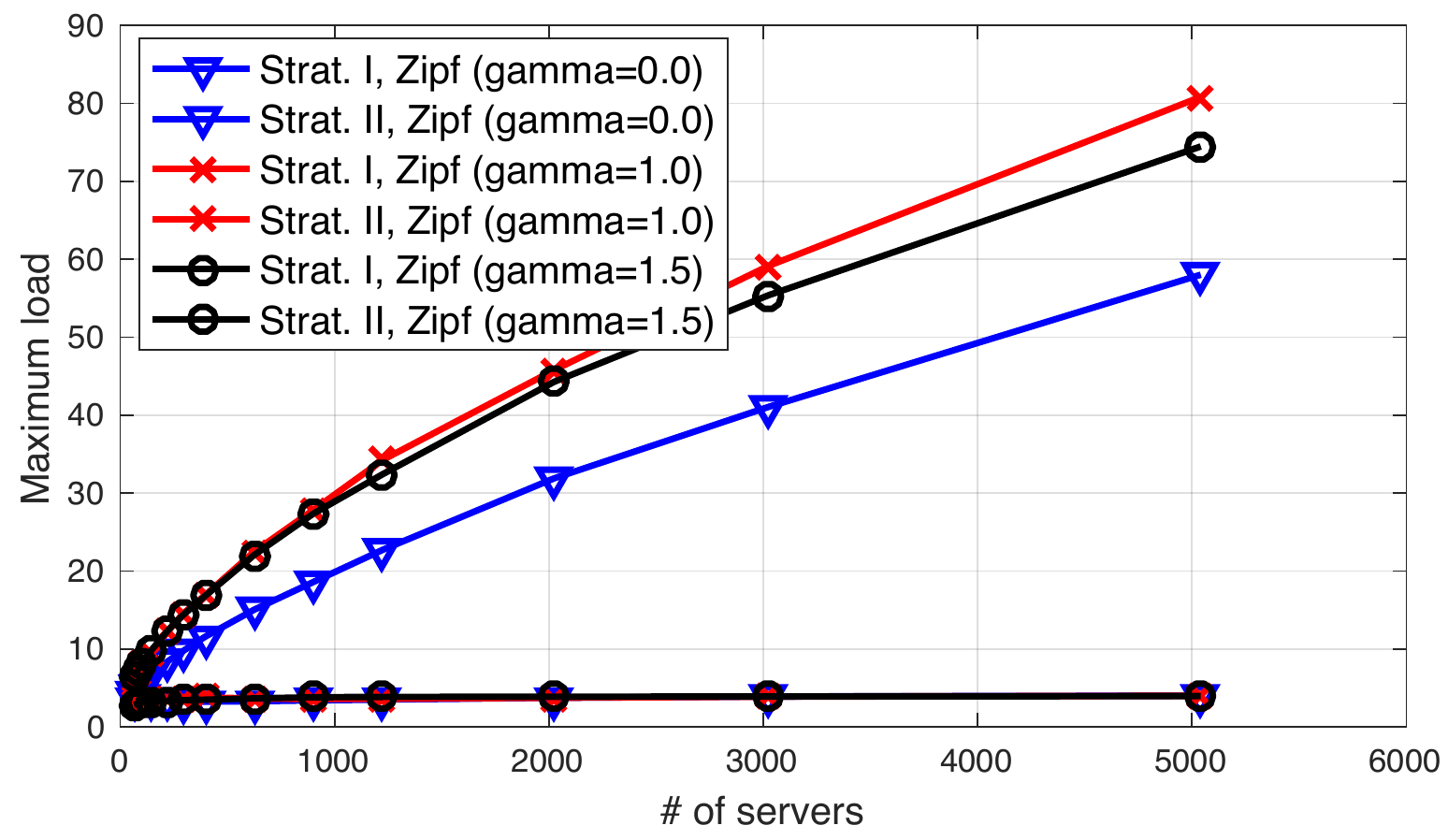}
\end{center}
\caption{The maximum load versus number of servers for Strategies~I and II. 
Here, the network topology is power law random graph. The remaining settings are similar to that of Figure~\ref{fig:SrvSzVar_Lattice_Zipf_fn=64_cs=2_itr=1000}.}
\label{fig:SrvSzVar_BarabasiAlbert_Zipf_1ch_vs_2ch_fn=64_cs=2_itr=1000}
\end{figure}

Also Figures~\ref{fig:Tradeoff_Lattice_Zipf_srvn=2025_fn=500_cs=10_itr=5000},  \ref{fig:Tradeoff_RGG_Zipf_srvn=2025_fn=500_cs=10_itr=5000}, and \ref{fig:Tradeoff_BarabasiAlbert_Zipf_srvn=2025_fn=500_cs=10_itr=5000},
demonstrate the performance trade-off between maximum load and communication cost 
for different network topologies (namely, torus, random RGG, and power law model) and  Zipf distribution with parameter $\gamma\in\{0,1,1.5, 2\}$.
All these simulations show that the trends and the trade-off we found in our theoretical results are also valid for more practical network settings.

For convenience, a summary of simulation parameters are stated in Table~\ref{tab:Sum_Sim_for_Figures}.

\begin{figure}
\begin{center}
\includegraphics[width=0.48\textwidth]{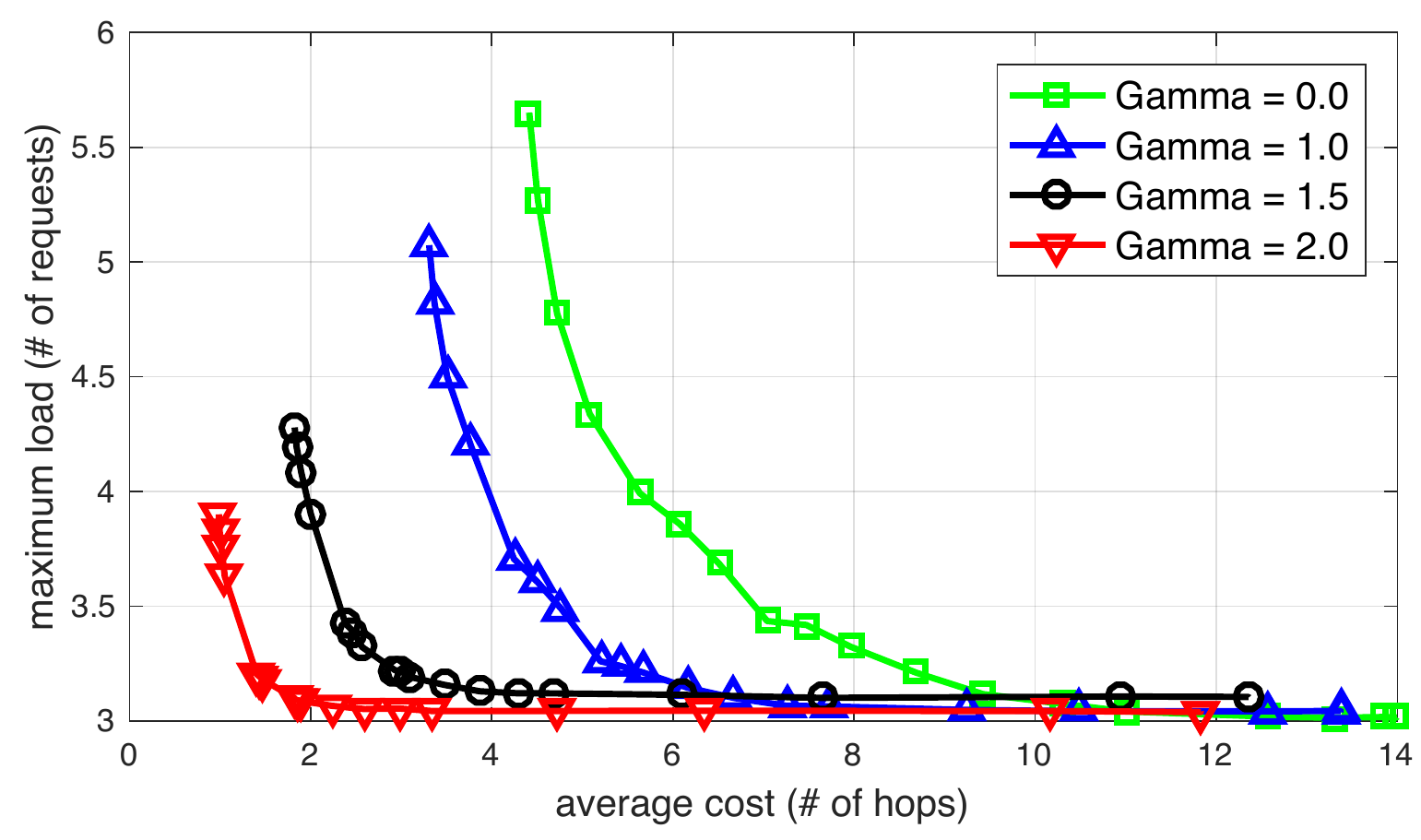}
\end{center}
\caption{The trade-off between the maximum load and communication cost for Strategy~II.
Here, the network topology is a torus of size $2025$, the file popularity is Zipf, the library size is $K=500$, and the cache size is $M=10$. Each point is an average over $5000$ independent simulation runs.}
\label{fig:Tradeoff_Lattice_Zipf_srvn=2025_fn=500_cs=10_itr=5000}
\end{figure}

\begin{figure}
\begin{center}
\includegraphics[width=0.48\textwidth]{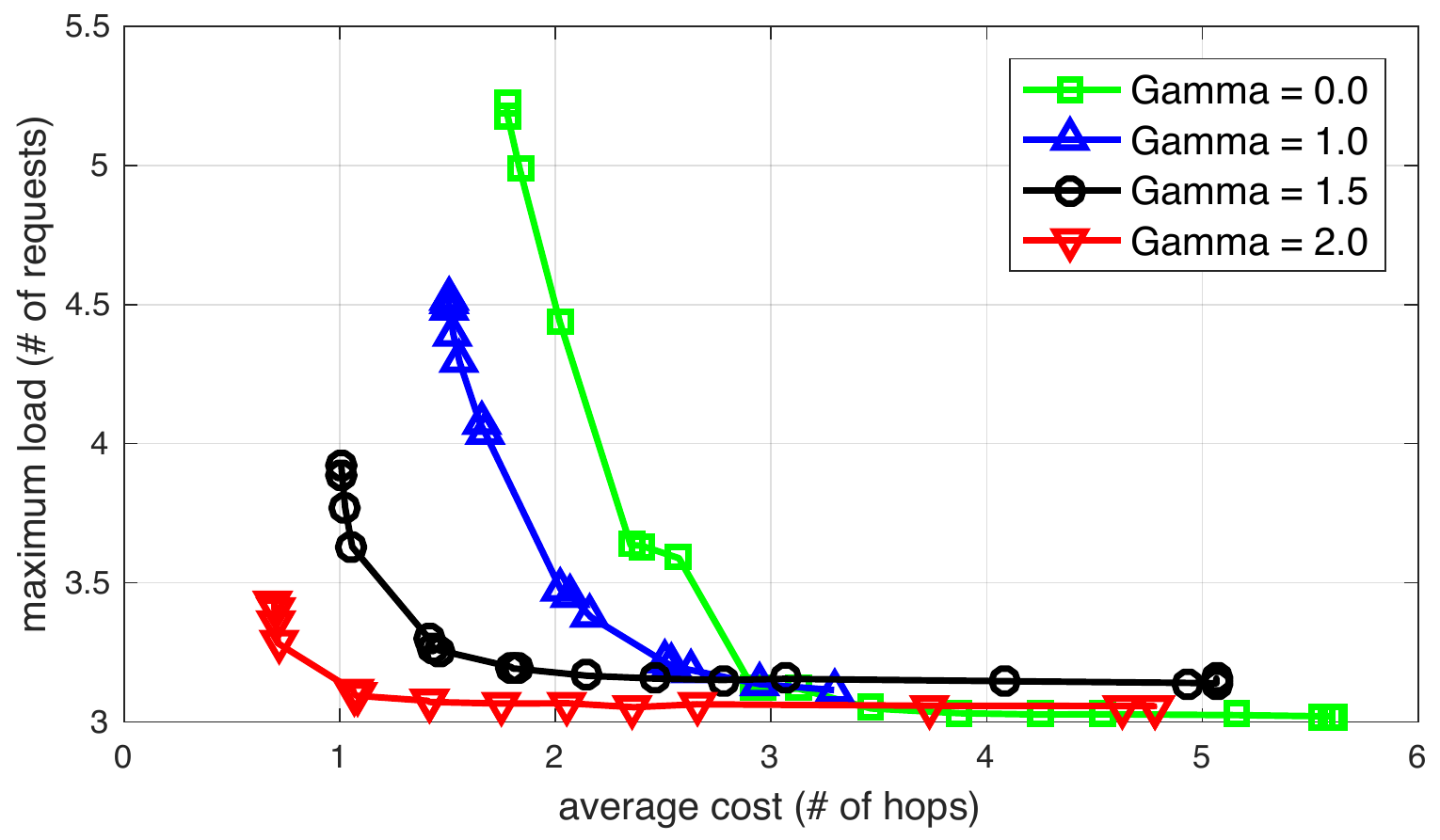}
\end{center}
\caption{The trade-off between the maximum load and communication cost for Strategy~II.
Here, the network topology is a RGG of size $2025$. The remaining settings are similar to that of Figure~\ref{fig:Tradeoff_Lattice_Zipf_srvn=2025_fn=500_cs=10_itr=5000}.}
\label{fig:Tradeoff_RGG_Zipf_srvn=2025_fn=500_cs=10_itr=5000}
\end{figure}

\begin{figure}
\begin{center}
\includegraphics[width=0.48\textwidth]{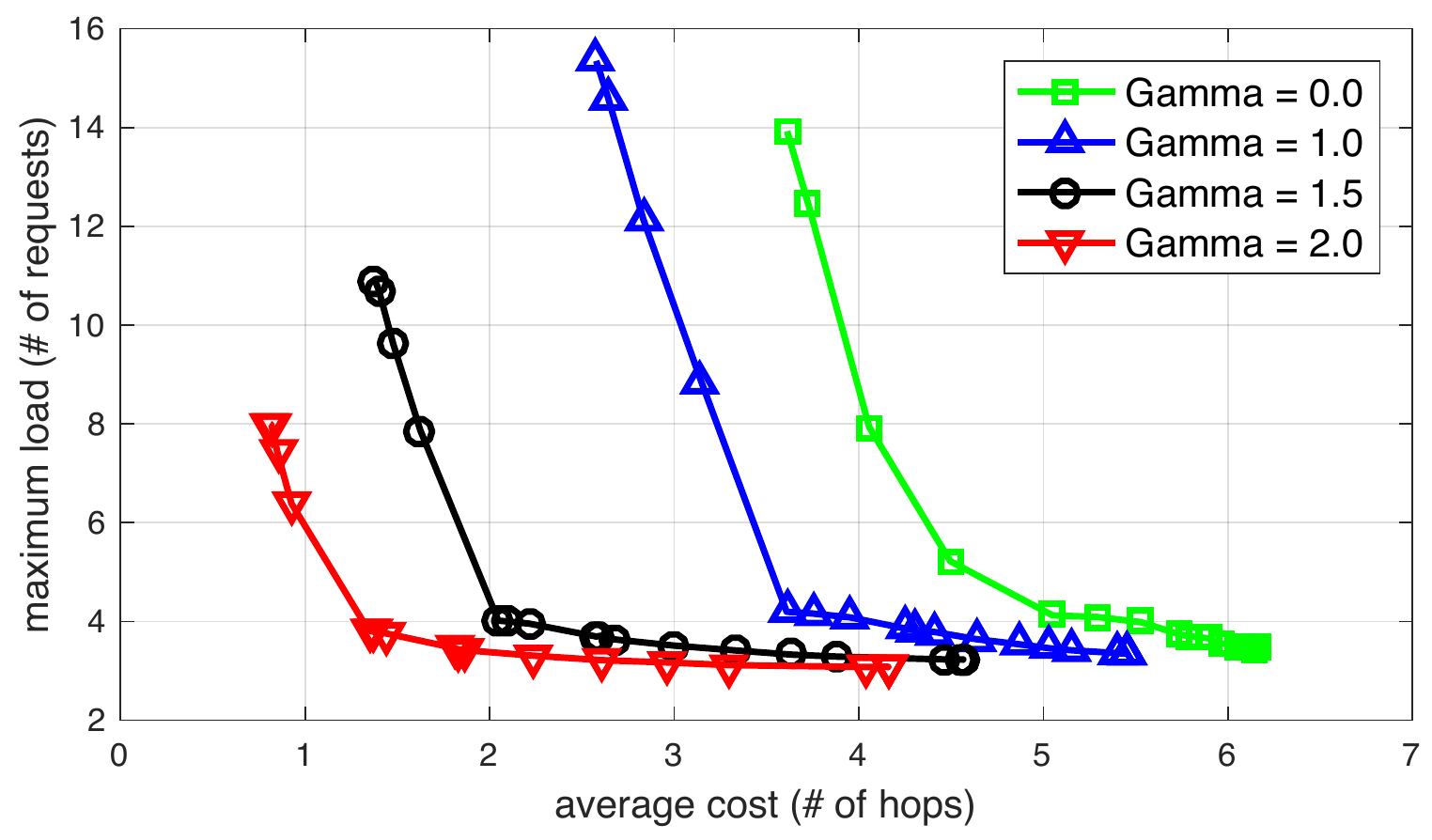}
\end{center}
\caption{The trade-off between the maximum load and communication cost for Strategy~II.
Here, the network topology is a power law random graph of size $2025$. The remaining settings are similar to that of Figure~\ref{fig:Tradeoff_Lattice_Zipf_srvn=2025_fn=500_cs=10_itr=5000}.}
\label{fig:Tradeoff_BarabasiAlbert_Zipf_srvn=2025_fn=500_cs=10_itr=5000}
\end{figure}

\renewcommand{\arraystretch}{1.2}
\begin{table}
\begin{center}
\begin{tabular}{|c|c|c|c|}
\hline
Fig. Number & Net. Topology & Strategy & Popularity\\
\hline
\ref{fig:OneChoice_SrvSzVar_MaxLoad}, \ref{fig:OneChoice_CacheSzVar_AvgCost} & Torus & Strategy I & Uniform\\
\hline
\ref{fig:TwoChoice_SrvSzVar_MaxLoad}, \ref{fig:TwoChoice_SrvSzVar_AvgCost} & Torus & Strategy II ($r=\infty$) & Uniform\\
\hline
\ref{fig:Lattice-SrvSzVar-1ch_vs_2ch} & Torus & Strategy I, II ($r=\infty$) & Uniform\\
\hline
\ref{fig:Tradeoff_srvn=2025_fn=500_itr=5000} & Torus & Strategy II ($r$ varying) & Uniform\\
\hline
\ref{fig:SrvSzVar_Lattice_Zipf_fn=64_cs=2_itr=1000} & Torus & Strategy I, II ($r=\infty$) & Zipf\\
\hline
\ref{fig:SrvSzVar_RGG_Zipf_1ch_vs_2ch_fn=64_cs=2_itr=1000} & RGG & Strategy I, II ($r=\infty$) & Zipf\\
\hline
\ref{fig:SrvSzVar_BarabasiAlbert_Zipf_1ch_vs_2ch_fn=64_cs=2_itr=1000} & Power Law & Strategy I, II ($r=\infty$) & Zipf\\
\hline
\ref{fig:Tradeoff_Lattice_Zipf_srvn=2025_fn=500_cs=10_itr=5000} & Torus & Strategy II ($r$ varying) & Zipf\\
\hline
\ref{fig:Tradeoff_RGG_Zipf_srvn=2025_fn=500_cs=10_itr=5000} & RGG & Strategy II ($r$ varying) & Zipf\\
\hline
\ref{fig:Tradeoff_BarabasiAlbert_Zipf_srvn=2025_fn=500_cs=10_itr=5000} & Power Law & Strategy II ($r$ varying) & Zipf\\
\hline
\end{tabular}
\end{center}
\caption{Summary of the simulation parameters for each figure.}
\label{tab:Sum_Sim_for_Figures}
\end{table}

\section{Discussion and Concluding Remarks}
\label{sec:Discussion_Remarks}
In this section, we first discuss three important practical issues related to our proposed scheme, then we will conclude the paper.

Our theoretical results in Sections \ref{sec:Nearest_Replica_Strategy} and \ref{sec:Proximity-Aware_Two_Choices_Strategy} are stated for a 2D-Grid topology. The main reason for assuming this rather unrealistic topology is developing the main idea of the paper clearly. However, it should be noticed that our approach can be extended to more general graph models at the expense of lengthy proofs and calculations. For example, as mentioned in Section \ref{sec:Proximity-Aware_Two_Choices_Strategy}, $|B_r(u)|$ is the size of the ball of radius $r$ around node $u$. The main feature of the 2D-Grid which affects our results is that $|B_r(u)|=\Theta(r^2)$ for all $u$. Now suppose that, instead of assuming a 2D-Grid, we consider a graph in which $|B_r(u)|=\Theta(r^{\mathrm{dim}})$, w.h.p. Then the parameter $\mathrm{dim}$ will appear in our results instead of $\mathrm{dim}=2$ in the special case of 2D-Grid. For example, the term $\alpha+2\beta$ in  the statement of Theorem 4 would be generalized to $\alpha+\mathrm{dim}\times\beta$. More generally, even for $|B_r(u)|=\Theta(f(r))$, our basic technicalities can be extended too. 
However, in this paper we investigate other network topologies, such as Random Geometric Graphs and Scale-Free (power law) networks via extensive simulations in Section \ref{sec:Simulations}. As we have discussed it there, the main trends are also valid for these topologies as well.
Thus, our findings cover a wider class of graph models which are more similar to real-world CDN network topologies.

The proposed \emph{proximity-aware two choices} scheme can be implemented in a distributed manner. To see why, notice that upon arrival of each request at each server, this strategy needs two kinds of information to redirect the request. This information can be provided to the requesting server without the need for a centralized authority in the following way. The first one is the cache content of other users in its neighborhood with radius $r$. Since, the cache content dynamic of servers is much slower than the requests arrival, this can be done by periodic polling of nearby servers without introducing much overhead (also see Distributed Hash Table (DHT) schemes, \eg, \cite{BauerHW07} and \cite{KargerLLPLL97}). The second type of information is the queue length information of two randomly chosen nodes inside its neighborhood with radius $r$, which can also be efficiently done in a distributed manner by polling or piggybacking.

In practice, request arrivals and servers' operation happen in continuous time which needs a queuing theory based performance analysis. However, as shown in \cite{Mitzenmacher01} and \cite{sur01}, the behaviour of load balancing schemes in continuous time (\ie, known as the supermarket model) and static balls and bins problems are closely related. Thus, we conjecture that our proposed scheme will also have the same performance in the queuing theory based model. We postpone a rigorous analysis of such scenario to future work.

In summary, in this work, we have considered the problem of randomized load balancing and its tension with communication cost and memory resources in cache networks. By proposing two request assignment schemes, this trade-off has been investigated analytically. Moreover, extensive simulation results support our theoretical findings and provide practical design guidelines.

\section*{Acknowledgment}
The authors would like to thank Seyed~Abolfazl~Motahari, Omid Etesami, Thomas~Sauerwald and Farzad~Parvaresh for helpful discussions and feedback.

\bibliographystyle{IEEEtran}
\bibliography{diss1}

\appendices

\section{Some Tail Bounds}\label{app:bounds}

\begin{thm}[Chernoff Bounds]\label{app:cher}
	Suppose that $X_1, X_2,\ldots, X_n\in \{0, 1\}$ are independent random variables and let $X=\sum_{i=1}^n X_i$. 
	Then for every $\delta\in (0, 1)$ the following inequalities hold:
	\begin{align*}
	\Pr{X\ge (1+\delta)\Ex{X}} &\le \exp(-\delta^2 \Ex{X}/2),\\
	\Pr{X\le (1-\delta)\Ex{X}} &\le \exp(-\delta^2 \Ex{X}/3).
	\end{align*}
	In particular,
	\[
	\Pr{|X-\Ex{X}|\ge \delta\Ex{X}}\le 2\exp(-\delta^2 \Ex{X}/3).
	\]
\end{thm}
For a proof see \cite{DP09}.

To deal with moderate independency we can state the following lemma. 
\begin{lemma}[Deviation bounds for moderate independency, see  {\cite[Lemma 1.18]{doerrbook}}]
\label{mod-chernoff}
	Let $X_1,\ldots, X_n$ be arbitrary binary random variables. Let
	$X_1^*, X_2^*, \ldots, X_n^*$
	be binary random variables that are mutually independent and
	such that for all $i$, $X_i$
	is independent of $X_1, \ldots, X_{i-1}$. Assume that for all
	$i$ and all $x_1, \ldots, x_{i-1} \in \{0, 1\}$,
	\[
	\Pr{X_i = 1|X_1= x_1,\ldots, X_{i-1}= x_{i-1}}\ge \Pr{X^*_i= 1}.
	\]
	Then for all $k \ge 0$, we have
	\[
	\Pr{ \sum_{i=1}^nX_i\le k}\le \Pr{\sum_{i=1}^n X^*_i\le k}
	\]
	and the latter term can be bounded by any deviation bound for independent
	random variables.
\end{lemma}

\end{document}